\documentclass[11pt]{article}
\usepackage{amsthm,collect,fullpage}

\usepackage{booktabs}
\usepackage{subcaption}

\usepackage{amsmath,amssymb,amsfonts}

\usepackage[citecolor=blue,linkcolor=blue,colorlinks=true]{hyperref}
\usepackage{xcolor}
\usepackage{caption}
\usepackage{verbatim,graphicx}
\usepackage{complexity}
\usepackage[capitalise] {cleveref}	
\usepackage{xspace}

\usepackage{framed}
\usepackage{multicol}
\usepackage{todonotes}

\theoremstyle{plain}
\newtheorem{theorem}{Theorem}[section]
\newtheorem{lemma}[theorem]{Lemma} 
\newtheorem{proposition}[theorem]{Proposition}
\newtheorem{claim}[theorem]{Claim}

\newtheorem{question}[theorem]{Question}

\theoremstyle{definition}
\newtheorem{definition}[theorem]{Definition}

\newcounter{todocounter}

\newcommand{\sse}{\subseteq}
\newcommand{\zo}{\{0,1\}}
\newcommand{\zon}{\zo^n}
\newcommand{\defn}{\stackrel{\text{\tiny def}}{=}}

\newcommand{\set}[1]{\left\{ #1 \right\}}
\newcommand{\mytilde}[1]{\overset{\sim}{#1}}

\newcommand{\etal}{\textit{et al}.\@\xspace}
\newcommand{\ie}{i.e.}

\newcommand{\DT}{\mathsf{DT}}

\newcommand{\psens}{\mathsf{psens}}
\newcommand{\psenss}{\mathsf{\widetilde{psens}}}
\newcommand{\EC}{\mathsf{EC}}
\newcommand{\negs}{\mathsf{negs}}
\newcommand{\fEC}{\mathsf{EC^{F}}}
\newcommand{\mEC}{\mathsf{EC^{M}}}
\newcommand{\entm}{\mathsf{H_{max}}}

\newcommand{\KW}{\mathsf{KW}}
\newcommand{\stconn}{{\mathsf{ STCONN}}}

\newcommand{\calB}{{\cal B}}
\newcommand{\calP}{{\cal P}}
\newcommand{\calT}{{\cal T}}

\newcommand{\msize}{\mathsf{mSize}}
\newcommand{\size}{\mathsf{Size}}
\newcommand{\depth}{\mathsf{Depth}}

\newcommand{\energy}{\mathsf{Energy}}

\usepackage{mathpazo}

\def\movetoappendix{0}

\definecollection{appendix}
\makeatletter
\newenvironment{aproof}[2]
  { \@nameuse{collect}{appendix}
  { \subsection{#1} \label{#2} \begin{proof} } {\end{proof}}
  }{\@nameuse{endcollect}}
\makeatother

\makeatletter
\newenvironment{appsection}[2]
  { \@nameuse{collect}{appendix}
  { \subsection{#1} \label{#2} }
  {}
  }{\@nameuse{endcollect}}
\makeatother

\ifthenelse{\equal{\movetoappendix}{0}}{

}{}

\makeatletter
\newtheorem*{rep@theorem}{\rep@title} \newcommand{\newreptheorem}[2]{
\newenvironment{rep#1}[1]{
        \def\rep@title{{\textbf{#2}} \ref{##1}}
        \begin{rep@theorem} }
        {\end{rep@theorem} } }
        \makeatother

\newreptheorem{theorem}{Theorem}
\newreptheorem{lemma}{Lemma}
\newreptheorem{prop}{Proposition}

\newcommand{\footremember}[2]{
	\footnote{#2}
	\newcounter{#1}
	\setcounter{#1}{\value{footnote}}
}
\newcommand{\footrecall}[1]{
	\footnotemark[\value{#1}]
}
\title{Characterization and Lower Bounds for Branching Program Size using Projective Dimension}
\author{Krishnamoorthy Dinesh\footremember{a}{Indian Institute of Technology Madras, Chennai, India.
		(\texttt{\{kdinesh,sajin,jayalal\}@cse.iitm.ac.in})} \and Sajin Koroth\footrecall{a} \and  Jayalal Sarma\footrecall{a}
}

\title{New Bounds for Energy Complexity of Boolean Functions\footnote{ Preliminary version appeared in 24th International Computing and Combinatorics Conference (COCOON 2018).}}
\author{Krishnamoorthy Dinesh\footremember{a}{Indian Institute of Technology Madras, 
Chennai, India. Email: {\{\tt kdinesh,jayalal\}@cse.iitm.ac.in}}  \and 
Samir Otiv\thanks{Maximl Labs, India. 
Part of the work was done while this author was a student at the Indian 
Institute of Technology Madras.
Email: {\tt  samir.otiv@maximl.com}}  \and Jayalal Sarma\footrecall{a} } 
\date{}
\begin{document}
\maketitle
\begin{abstract}
For a Boolean function $f:\zon \to \zo$ computed by a Boolean circuit $C$ over a 
finite basis $\calB$, the \textit{energy complexity} of $C$ (denoted by $
\EC_{\calB}(C)$) is the maximum over all inputs $\zon$ of the number gates 
of the circuit $C$ (excluding the inputs) that output a one. Energy 
Complexity of a Boolean function over a finite basis $\calB$ denoted by $
\EC_\calB(f)\defn \min_C \EC_{\calB}(C)$ where $C$ is a Boolean circuit over $\calB
$ computing $f$.

We study the case when $\calB = \{\land_2, \lor_2, \lnot\}$, the standard
Boolean basis. It is known that any Boolean function can be computed by a
circuit (with potentially large size) with an energy of at most
$3n(1+\epsilon(n))$ for a small $ \epsilon(n)$(which we observe is improvable
to $3n-1$). We show several new results and connections between energy
complexity and other well-studied parameters of Boolean functions.
\begin{itemize}
\item For all Boolean functions $f$, $\EC(f) \le O(\DT(f)^3)$ where $\DT(f)$ is the optimal decision tree depth of $f$. 
\item We define a parameter \textit{positive sensitivity} (denoted by
$\psens$), a quantity that is smaller than sensitivity [Cook \etal, SIAM Journal of Computing, 15(1):87--97, 1986]
and defined in a similar 
way, and show that for any Boolean circuit $C$ computing a Boolean 
function $f$, $ \EC(C) \ge \psens(f)/3$.
\item For a monotone function $f$, we show that $\EC(f) = \Omega(\KW^+(f))$ where $\KW^+(f)$ is the cost of monotone Karchmer-Wigderson game of $f$.
\item Restricting the above notion of energy complexity to Boolean formulas,  we show $\EC(F) = \Omega\left (\sqrt{\L(F)}-\depth(F)\right )$ where $\L(F)$ is the size and $\depth(F)$ is the depth of a formula $F$.
\end{itemize}
\end{abstract}
\section{Introduction}
For a Boolean function $f:\zon \to \zo$ computed by a Boolean circuit $C$ over a basis
$\calB$, the \textit{energy complexity} of $C$ (denoted by $\EC_{\calB}(C)$) is
the maximum over all inputs $\zon$ the numbers of gates of the circuit $C$
(excluding the inputs) that outputs a one. The energy complexity of a Boolean
function over a basis $\calB$ denoted by $\EC_\calB(f)\defn \min_C
\EC_{\calB}(C)$ where $C$ is a Boolean circuit over $\calB$ computing $f$. A
particularly interesting case of this measure of Boolean function, is when the
individual gates allowed in the basis $\mathcal{B}$ are threshold gates (with
arbitrary weights allowed). In this case, the term energy in the above model
captures the number of neurons firing in the cortex of the human brain (see
\cite{UDM06} and the references therein).  This motivated the study of upper
and lower bounds \cite{UDM06} on various parameters of energy efficient
circuits - in particular the question of designing threshold circuits which are
efficient in terms of energy as well as size computing various Boolean
functions.

Indeed, irrespective of the recently discovered motivation mentioned above, the
notion of energy complexity of Boolean functions, has been studied much before.
Historically, the measure of energy complexity of Boolean functions\footnote{A
related notion has been studied in~\cite{Glo82} where the energy is the number
times the gates in a circuit switches its value. Recent
studies~\cite{ABNPS14,NMKM15} looks at the energy of a circuit as a function of
the voltage applied to the gates thereby allowing some gates to fail. 
We remark that the notion of energy of Boolean circuits studied in this paper
is very different from those studied in the works mentioned.} was first studied
by Vaintsvaig~\cite{V62} (under the name ``power of a circuit''). Initial
research was aimed at understanding the maximum energy needed to compute any
$n$ bit Boolean function for a \textit{finite} basis $\mathcal{B}$ (denoted by
$\EC_{\calB}(n)$). Towards this end, Vaintsvaig~\cite{V62} showed that for any
finite basis $\mathcal{B}$, the value of $\EC_{\calB}(n)$ is asymptotically
between $n$ and $\frac{2^n}{n}$. Refining this result further,
Kasim-zade~\cite{K92} gave a complete characterization by showing the
following remarkable trichotomy: for any finite complete basis $\calB$, either
$\EC_\calB(n) = \Theta(2^n/n)$ or $ \Omega(2^{n/2}) \le \EC_\calB(n) \le
O(\sqrt{n}2^{n/2})$ or $\Omega(n) \le \EC_\calB(n) \le O(n^2)$.

An intriguing question about the above trichotomy is where exactly does the standard 
Boolean basis $\calB=\{ \land_2, \lor_2, \lnot \}$ fits in. By an explicit Boolean circuit 
construction, Kasim-zade~\cite{K92} showed that $\EC_
\calB(n) \le O(n^2)$. 
Recently, Lozhkin and Shupletsov~\cite{LS15} states (without proof) 
that the Boolean circuit construction by Kasim-zade~\cite{K92} over the complete Boolean 
basis is of energy $4n$, thus deriving that $\EC_\calB(n) \le 4n$. 
Lozhkin and Shupletsov improves it to $3n(1+\epsilon(n))$ by constructing a Boolean circuit 
of size 
$\frac{2^n}{n}(1+\epsilon(n))$ for an $\epsilon(n)$ tending to $0$ for large $n$. 
We observe that this bounds can be further improved to be at most $3n-1$ while the 
size is $2^{O(n)}$
by carefully following the construction in~\cite{LS15}~(\cref{ub:fanin-bounded}).

As mentioned in the beginning, in a more recent work, for the case when the basis 
is threshold gates\footnote{With values of the weights and threshold being 
arbitrary rational numbers, notice that this basis is no longer finite and hence 
the bounds and the related trichotomy are not applicable.}, Uchizawa~\etal~\cite{UDM06} initiated the study of energy complexity for threshold circuits. 
More precisely, they defined the energy complexity of threshold circuits and gave 
some sufficient conditions for certain functions to be computed by small energy 
threshold circuits. In a sequence of works, Uchizawa \etal~\cite{UNT10,UTN11} 
related energy complexity of Boolean functions under the threshold basis to the 
other well-studied parameters like circuit size and depth for interesting classes 
of 
Boolean functions. In a culminating result, Uchizawa and Takimoto~\cite{UT08} 
showed that constant depth thresholds circuits of unbounded weights with the 
energy restricted to $n^{o(1)}$ needs exponential size to compute the Boolean inner 
product function\footnote{$\IP(x,y) = \sum_i x_i y_i \mod 2$}. 
This is also important in the context of circuit lower bounds, 
where it is an important open question to prove exponential lower bounds against 
constant depth threshold circuits in general (without the energy constraints) for 
explicit functions.\\[-2mm]

\noindent{\bf Our Results:} Returning to the context of standard Boolean basis $
\calB = \{\land_2, \lor_2, 
\lnot\}$, we show several new results and connections between energy complexity and 
other Boolean function parameters. Since we are interested only in the standard 
Boolean basis $\calB$, we use $\EC(f)$ to denote $\EC_{\calB}(f)$.\\[-2mm]

\noindent {\bf Upper bounds for Energy Complexity:} As our first and 
main contribution, we show new bounds on energy complexity  of Boolean functions by 
two other parameters of functions, 
one in terms of an upper bound and the other in terms of a lower bound.

For a function $f:\zon \to \zo$, let $\DT(f)$ denote the decision tree 
complexity of the Boolean function - the smallest depth of any decision tree 
computing the function $f$. We state our main result: 
\begin{theorem}[Main]\label{energy:dt}
For any Boolean function $f$, $\EC(f) \le O(\DT(f)^3)$.
\end{theorem}
We remark that the size of the Boolean circuit constructed above is exponential in $
\DT(f)$. There are several Boolean functions for which the decision trees are very shallow - a demonstrative example is the \textit{tree function} (see~\cref{sec:prelims} for a definition) computable by a decision tree of depth $\log n$. Our result implies that there is a Boolean circuit computing this function of energy $O(\log^3 n)$.

In terms of the energy of the circuit, this improves the 
bounds of Lozhkin and Shupletsov~\cite{LS15} when $\DT(f)$ is strictly smaller than $\sqrt[3]{n}$.

On a related note, Uchizawa~\etal~\cite{UDM06}, as a part of their main proof, 
showed a similar result for threshold decision trees which are decision trees 
where each internal node can query an arbitrary weighted threshold function on 
input variables. Let $\DT_{th}(f)$ denotes the depth of the smallest depth threshold 
decision tree computing $f$. For a basis $\calT$ consisting of arbitrary threshold 
functions, their results implies that $\EC_{\calT}(f) \le 1+ \DT_{th}(f)$ (see~
\cref{th-dtlb} for details). 
Since their construction produces a weighted threshold circuit, it does not 
directly give us a low energy Boolean circuit even for Boolean decision trees.\\[-2mm]

\noindent {\bf Update after our work~\cite{DOS18} in connection with \cref{energy:dt}:} Recently (after the conference version of our work~\cite{DOS18} was published), Sun~\etal~\cite{SSWX18} improved the upper bound for $\EC(f)$ in~\cref{energy:dt} from $O(\DT(f)^3)$ to $O(\DT(f)^2)$. In addition, they also showed that $\EC(f) = \Omega(\sqrt{\DT(f)})$ thereby showing that the two parameters are polynomially related.\footnote{See~\cref{ec:lb:dt} for a discussion and comparison of techniques for proving lower bound on energy complexity.} While this improves our main result of the conference version of the paper~\cite{DOS18} in terms of the relationship between energy complexity and decision tree complexity of a Boolean function, in terms of the lower bound that we can obtain for energy complexity, the improvement really depends on the function family considered. We demonstrate this in a comparison between the lower bound methods in Section~\ref{ec:lb:dt}.\\[-2mm]

\noindent {\bf Lower Bounds for Energy Complexity :} 
To obtain lower bounds on energy, we define a new parameter called the positive 
sensitivity (which is at most the sensitivity of the Boolean function~\cite{cdr86}). Let $[n] \defn \set{1,\ldots,n}$. For a 
function $f:\zon \to \zo$ and an input $a \in \zon$, we define the {\em positive 
sensitivity} of $f$ on $a$ (denoted by $\psens(f,a)$) as the number of indices $i\in[n]$ such that $a_i=1$ and flipping the bit $a_i$ causes the function to change its value. We define $\psens(f)$ to be $\max_{a \in \{0,1\}^n}\psens(f,a)$.
Using this parameter, we show the following.
\begin{theorem}\label{psens:energy}
For any Boolean function $f:\zon \to \zo$ computed by a Boolean circuit $C$, $ \EC(C) \ge 
\psens(f)/3$.
\end{theorem}

The main tool in proving the above results is the 
notion of \textit{continuous positive paths} which are paths in a Boolean circuit where all 
the gates in the path evaluate to $1$. Using the same tool, we show that the 
monotone Karchmer-Wigderson games can be 
solved by exchanging at most $\EC(C) \log c$ where $C$ is a Boolean circuit with 
fan-in at most $c$ (see \cref{lb:kw:energy} for more details). This implies the 
following energy lower bound for computing  monotone functions.

\begin{theorem}\label{kw:lb}
Let $f:\zon \to \zo$ be a monotone function. Then $\EC(f) = \Omega(\KW^+(f))$.
\end{theorem} 
It is known that for the perfect matching function of a graph on $n$ edges, denoted 
as $f_{PM}$, $\KW^+(f_{PM}) = \Omega(\sqrt{n})$~\cite{RW92}. Hence,~\cref{kw:lb} implies 
that any Boolean circuit with bounded fan-in, computing $f_{PM}$ will require 
energy at least $\Omega(\sqrt{n})$.

All the models considered so far are of fan-in $2$. We now relax this requirement 
and consider the energy complexity of unbounded fan-in constant depth Boolean circuits 
computing specific functions. In this direction, we show the following.

\begin{theorem}\label{lb:depththree:raz} 
Let $C$ be any unbounded fan-in Boolean circuit of depth $3$ computing the parity function 
on $n$ variables. Then, $\EC(C)$ is $\Omega(n)$.
\end{theorem}

Finally, we show lower bounds on the energy complexity of Boolean functions 
when restricted to Boolean formulas (instead of Boolean circuits), in terms of its formula size and 
depth.

For a formula $F$, let $\L(F)$ be the number of leaves in $F$ and $\depth(F)$ be 
the length of the longest path from root to any leaf in $F$. For a Boolean function 
$f$, let $\L(f)$ be the minimum $\L(F)$ among all the formulas $F$ computing $f$. 
Let $\fEC(f)$  be 
the minimum energy for any bounded fan-in formula computing $f$. 
Intuitively, Boolean formulas can take more energy than a Boolean circuit since we cannot 
``reuse" computation. Also, for any formula $F$, $\fEC(F)\le \L(F)-1$. Hence, it 
would not be surprising if $\fEC(F)$ is also lower bounded by $\Omega(\L(F))$ giving 
a tight bound of $\fEC(f) = \Theta(\L(f))$. Towards this direction we show the 
following result.

\begin{theorem}\label{energy:formula}
For a Boolean function $f$, computed by a formula $F$, $$ \fEC(F) = \Omega 
\left(\sqrt{\L(F)}-\depth(F)\right ).$$
\end{theorem}

\noindent
{\large\textbf{Related work:}} 
We discuss recent results on energy complexity of computing Boolean 
functions in various circuit models.

Observe that any Boolean circuit is also a threshold circuit since each of gates in $\calB$, $\land$, $\lor$ and $\neg$ can be implemented by threshold gates.  This implies that, $\EC(f) \ge 
\EC_{\calT}(f)$. Hence, for a function $f$, known lower bound on $\EC_{\calT}(f)$ 
translates to a lower bound on $\EC(f)$. In this context,  \cref{tab:tradeoff} 
summarizes known results on bounds on energy complexity of threshold circuits in 
terms of the parameters size, depth and fan-in for certain classes of Boolean 
functions.  For designing energy efficient circuits, techniques or tools to reduce 
the energy complexity of circuits is relevant in this context. \cref{tab:transform} 
summarizes known results on energy complexity of Boolean 
functions on ways to transform circuits to energy efficient ones.

\paragraph{Energy vs circuit parameters:}
\begin{table}[htp!]
\centering
\begin{tabular}{@{}l|c|c|lcl@{}}
\toprule
{\tt Param}   & {\tt Function $f$ is }\ldots & {\tt Gate} & Trade-off                                                                                             & Ref. \\ \midrule
$\ell$          & Symmetric     & any          & $\ell \ge \frac{n-b_f}{e}$   & \cite{SUZ13}     \\
$s$           & Symmetric     & Unate        & $s^e \ge \frac{n+1-a_f}{b_f} $                                                       & \cite{UTN11}      \\
$d, s$ & any           & Threshold    & $R_{0.5-\delta}(f) = O(e^d \log s)$,  &       \cite{UT08}   \\ 
& & & $\delta = \frac{1}{s^{O(e^d)}}$  & \\
\bottomrule
\end{tabular}
\caption{Known bounds on energy $e$ of circuits computing Boolean functions}
\label{tab:tradeoff}
\end{table}

 \cref{tab:tradeoff} presents the information: {\it ``Energy vs Parameter (\texttt{Param}) trade-off for 
any circuit $C$ using specific type of gates (\texttt{Gate}) computing the  function $f$''}. The 
parameters involved are $s= \size(f)$, $d = \depth(f) $, $\ell$ =  fanin of gates 
in $C$ and $e$ is the optimum energy of a circuit with gates of type {\tt Gate} 
computing $f$.

By $R_{\delta}(f)$, we denote the two-sided error public coin 
randomized communication complexity of $f$ with error probability $\delta$.
We now describe the two notations $a_f$ and $b_f$ used in first two entries of~\cref{tab:tradeoff} for a symmetric function $f$.
Observe that any symmetric function $f:\zon \to \zo$ can be 
completely described by an $n+1$ length Boolean vector -- $v_f$ as $f(x) = v_f(|x|)$ for all $x\in \zon$ where $|x|$ is the number of ones in $x$. For $b \in \zo$, let $g_b$ is maximum length of consecutive $b$'s in $v_f$. Then,  $a_f \defn \min\set{g_0,g_1}$ and $b_f \defn \max\set{g_0,g_1}$. The first two entries of~\cref{tab:tradeoff} gives an energy fan-in trade-off and size energy trade-off, respectively, for symmetric functions in terms of $a_f$ and $b_f$. The third entry of~\cref{tab:tradeoff} gives an energy size trade-off for constant depth threshold circuits for any Boolean function.

\paragraph{Energy of circuits under change of basis}
\cref{tab:transform} presents the information: {\it
``Given a circuit $C$ with $\energy(C) = e$ of with gates of type $A$ then, there 
exists a circuit $C'$ with gates of type $B$ computing the same function as $C$ 
with bounds on $\size(C')$, $\depth(C')$, $\energy(C')$.''
}

\begin{table}[htp!]
\centering
\begin{tabular}{@{}llllll@{}}
\toprule
$A$ & $B$ & $\size(C')$ & $\depth(C')$ & $\energy(C')$ & Ref.\\
\midrule
Any & Threshold & $\le O((e+n)\size(C))$ & $O(e)$ &  -& \cite{UNT10} \\
Threshold/Unate & Threshold/Unate & $\le 2\cdot e \cdot \size(C)+1$ & $\le 2 \cdot e +1$ & $e$ & \cite{UNT10}\\
\bottomrule
\end{tabular}

\caption{Transforming circuit of type $A$ to type $B$}
\label{tab:transform}

\end{table}

\noindent 
\textbf{Organization of the paper.} The rest of the paper is organized as 
follows. We start with preliminaries in~\cref{sec:prelims}. We show new bounds on energy complexity in 
terms of the decision tree depth in \cref{sec:energy:dt}. Then, we show two methods to obtain lower bounds on energy complexity in \cref{sec:psens:lb} and \cref{sec:kw:lb} using the notion of \textit{continuous positive paths} (introduced in~\cref{sec:psens}). In \cref{sec:depththree}, we show energy lower bounds for depth $3$ Boolean circuits computing a specific function. Following 
this, in~\cref{sec:formula}, we show energy lower bounds for Boolean formulas. 
In~\cref{ec:lb:dt}, we compare of our lower bound techniques with a recent improvement due to Sun~\etal~\cite{SSWX18}. We conclude in~\cref{sec:conc} outlining some directions for further exploration. 

\section{Preliminaries}\label{sec:prelims}

A Boolean 
circuit $C$ over the basis $\calB=\set{ \land_2, \lor_2,\neg}$ is a directed 
acyclic graph (DAG) with a root node (of out-degree zero), input gates labeled by 
variables (of in-degree zero) and the non-input gates (inclusive of root) labeled 
by functions in $\calB$. Define the size 
to be the number of non-input gates and, depth to be the length of the longest 
path from root to any input gate of the circuit $C$ denoted, respectively, as $
\size(C)$ and $\depth(C)$. A Boolean formula is a Boolean 
circuit where the underlying DAG is a tree. We call a negation gate that takes  
input 
from a variable as a \textit{leaf negation}. A Boolean circuit is 
said to be \textit{monotone} if it does not use any negation gates. A function 
is monotone if it can be computed by a monotone circuit. Equivalently, a function  
$f$ is monotone if $\forall~x,y \in \zon$, $x \prec y \implies f(x) \le f(y)$ 
where $x \prec y$ iff $x_i \le y_i$ for all $i\in [n]$. For a Boolean circuit $C$, $
\negs(C)$ denotes the number of NOT gates in the circuit $C$. 
Fix an arbitrary ordering among the gates of $C$. A \textit{firing pattern} of a 
circuit $C$ on a given input is the binary string of evaluation of the gates on the 
input as per the fixed ordering. The \textit{number of firing patterns of a 
circuit} $C$ is the number of distinct firing patterns for $C$ over all inputs.

 For $i \in [n]$, let $e_i$ denote the $n$ length 
Boolean vector with the $i^{th}$ entry alone as $1$. For an $a \in \zon$, $a \oplus 
e_i$ denotes the input obtained by flipping the $i^{th}$ bit of $a$. 
The positive sensitivity of $f$ on $a$, denoted by $\psens(f,a)$, is the number of $i\in[n]$ such that $a_i=1$ and $f(a \oplus e_i) \ne f(a)$. We define $\psens(f)$ as $\max_{a \in 
\set{0,1}^n} \psens(f,a)$. 

For a monotone function $f:\zon \to \zo$, $x \in f^{-1}(1)$ and $y \in f^{-1}(0)$, 
define $S^+_f(x,y) = \set{ i \mid  x_i =1, y_i =0, i \in [n]}$. The monotone 
Karchmer-Wigderson cost of $f$ (denoted by $\KW^{+}(f)$) is the optimal 
communication cost of the problem where Alice has $x$, Bob has $y$ and they have 
to find an $i \in [n]$ such that $i \in S_f^+(x,y)$. It is known that $\KW^+(f)$ 
equals the minimum depth monotone Boolean circuit computing $f$. For more details 
about this model, see~\cite{textbook}.  We now define the tree function - Definition 3.1 of \cite{DS18}.
\begin{definition}[Tree Function] 
Let $\mathcal{F} = \set{f_k\mid k \in \mathbb{N}}$ be a family of Boolean functions where for every $k 
\in \mathbb{N}$,  $f_k:\zo^{2^k-1} \to \zo$ is defined by the decision
tree which is a full binary tree of depth $k$ with each of the $2^k-1$ internal node
querying a distinct variable and each of the nodes at level $k$ have left leaf child
labeled $0$ and right leaf child labeled $1$.
\end{definition}

\paragraph{\large Energy Complexity:}
For a Boolean circuit $C$ and an input $a$, the {energy complexity} of $C$ on the 
input $a$ (denoted by $\EC(C,a)$) is defined as the number of non-input gates that 
output a 1 in $C$ on the input $a$. Define the \textit{energy complexity} of $C$ 
(denoted by $\EC(C)$) as $\max_{a} \EC(C,a)$. The energy complexity of a function 
$f$, (denoted by $\EC(f)$) is the energy of the minimum energy circuit over the 
Boolean basis $\calB$ computing $f$. 

As mentioned in the 
introduction, Lozhkin and 
Shupletsov~\cite{LS15} showed that $\EC(f) \le 3n(1+\epsilon(n))$ by constructing a 
Boolean circuit of size $\frac{2^n}{n}(1+\epsilon(n))$ where $\epsilon(n) \to 0$ as $n \to \infty$. 
Their idea is to construct a Boolean 
circuit  of low energy that outputs all product terms on  $n$ variables where each 
of them appears exactly once in a negated or unnegated form. We call such terms as 
\textit{minterms}. We slightly improve the above 
bound using the same idea by constructing a Boolean circuit of size 
$2^{O(n)}$.

\begin{proposition} \label{ub:fanin-bounded}
For any $f:\zon \to \zo$, $\EC(f) \le 3n-1$.
\end{proposition}
\begin{proof}
We show that all minterms in $n$ variables can be computed by a Boolean circuit of energy 
at most $2n-1$. Assuming this, to compute $f$, construct an $\lor$ formula on $2^n$ 
inputs of depth $n$ and connect the minterms on which $f$ evaluates to $1$ as the leaves of $
\lor$ (and the rest of the inputs as $0$). Since on any input, exactly one 
of the leaves will evaluate to $1$, there is only $1$ path to the output gate where 
all $\land$ gates evaluate to $1$. Hence, the overall energy complexity is at most 
$2n-1+n = 3n-1$. We construct a Boolean circuit of energy $2n-1$ to compute all 
minterms on $n$ variables.

Proof is by induction on $n$. Let $x_1,\ldots, x_n$ be the variables. 
For $n=1$, the Boolean circuit is $x_1$, $\neg x_1$ which has energy $1$. Hence, the base 
case holds.

By induction, we have constructed a circuit $C$ (on $n$ inputs and having $2^n$ 
outputs) 
computing all $2^n$ minterms on $x_1,\ldots,x_n$. We modify the Boolean circuit as follows 
: branch out each output gate into two (left and right branch). Connect the left 
(resp. right) branch output to $x_{n+1}$ (resp. $\neg x_{n+1}$) by an $\land$
gate. Note that out of all $2^{n+1}$ outputs created this way, exactly one of them 
will output 1 on any input. Also we have computed all $2^{n+1}$ minterms on $x_1,
\ldots, x_{n+1}$. 
The resulting circuit has an energy of $2n-1$ for circuit $C$ by induction plus $2$ 
due to the output and the negation gate of $x_{n+1}$. Hence, the overall energy is 
$2n+1 = 2(n+1)-1$. This completes the induction.
\end{proof}

The upper bound in \cref{ub:fanin-bounded} has been improved by Sun~\etal~\cite{SSWX18} from $3n-1$ to $3n-2$. 

Observe that in a Boolean circuit $C$, for the leaf negation gates, there is always an 
input where all of them output a $1$. For the non-leaf negation gates, irrespective 
of the input, either the negation gate or its input gate  will output a one. Due to 
this reason, we have,
\begin{proposition}\label{energy:negs}
For any circuit $C$,  $\EC(C) \ge \negs(C)$.
\end{proposition}

\paragraph{\large Model specific variants of energy complexity:} We now consider 
the notion of energy complexity for three other circuit models, namely monotone 
circuits, Boolean formulas and threshold circuits.

\noindent
\textbf{Energy Complexity and Monotone Boolean circuits:}
For a monotone Boolean function $f$, computed by a monotone Boolean circuit $C$, define $
\mEC(C)$ as the maximum over all the inputs the number of non-input gates that 
output a $1$. We define $\mEC(f)$ as $\min_C \mEC(C)$ where $C$ is a monotone 
circuit computing $f$. The following proposition gives an exact characterization 
for $\mEC(f)$.

\begin{proposition}\label{energy:mon}
For a monotone Boolean function $f$, let $\msize(f)$ denotes the size of the 
smallest monotone Boolean circuit computing $f$. Then, $\mEC(f) = \msize(f)$.
\end{proposition}
\begin{proof}
Let $C$ be a monotone Boolean circuit of minimum size computing $f$. Clearly, $\mEC(f) \le 
\mEC(C) \le \msize(f)$. Also, for any monotone circuit $C'$ computing $f$, on the 
input $1^n$, all the gates in $C'$ output a $1$ implying $\mEC(C') \ge \mEC(C',1^n) 
= \msize(C')$. In particular, for the monotone circuit $C''$ of minimum energy computing $f
$, $\mEC(f) = \mEC(C'') \ge \size(C'') \ge \msize(f)$. Hence, $\mEC(f) = \msize(f)
$.
\end{proof}

\noindent
\textbf{Energy Complexity and Threshold circuits:}
Let $\calT$ be a basis consisting of all weighted threshold functions. A threshold 
circuit is a Boolean circuit where the gates are from the basis $\calT$.
Uchizawa \etal~\cite{UDM06} introduced the notion of energy complexity of threshold 
circuits denoted by $\EC_{\calT}(C)$, again defined as the worst energy of the 
threshold circuit $C$ among all the inputs. Define $\EC_{\calT}(f)$ as $\min_C 
\EC_{\calT}(C)$ where $C$ is a threshold circuit over the basis $\calT$ computing 
$f$.

A \textit{decision tree} is a rooted tree with all the non-leaf nodes labeled by 
variables and leaves labeled by a $0$ or $1$. Note that every assignment to the 
variable in the tree defines a unique path from root to leaf in the natural way. A 
Boolean function $f$ is said to be 
computed by a decision tree if for every input $a$, the path from root to a leaf 
guided by the input is labeled by $f(a)$. Depth of a decision tree is the length of 
the longest path from root to any leaf. Define decision tree depth of $f$ 
(denoted by $\DT(f)$) as the depth of the minimum depth decision tree computing $f
$. A threshold decision tree is similar to the decision tree except that queries at 
each non-leaf node can be an arbitrary threshold function on the input variables. 
We denote the depth of the minimum depth threshold decision tree computing $f$ by $
\DT_{th}(f)$.

For an $f: \zon \to \zo$,
 Uchizawa \etal~\cite{UDM06} introduced a measure of energy for threshold decision 
 tree $T$ computing $f$, denoted by $cost(T)$ defined as the maximum over all paths 
from root to leaf, the number of right turns taken in a path with cost of leaf defined to be zero. As mentioned in~\cite{UDM06}, $cost(T)$ can be seen as a measure of how often does a threshold gate in $T$ output a $1$.

As a part of their main result they showed, given any threshold decision tree $T$, 
(1) how to construct another threshold decision tree $T'$ with a bound on $cost(T')
$ (Lemma 2, 3 of~\cite{UDM06}) 
 and (2) how to obtain a low energy threshold circuit $C'$ computing $f$ from $T$ of energy $cost(T)$.
 (Lemma 5 of~\cite{UDM06}). 

This implies the following relation between $\EC_{\calT}(f)$ and $\DT_{th}(f)$.
\begin{proposition}\label{th-dtlb}
For any Boolean function $f$, $\EC_{\calT}(f) \le \DT_{th}(f)+1$.
\end{proposition}
\begin{proof}
Let $T$ be an optimum threshold decision tree computing $f$ with depth  $\DT_{th}
(f)$. We first state the results (1) and (2) formally. Result (1) says that $f$ can 
be computed by another threshold decision tree $T'$ with same depth as $T$, same 
number of leaves as $T$ and have $cost(T')\le \log (\# \text{ leaves of } T)$ 
(Lemma 2, 3 of~\cite{UDM06}). Result (2) says that there exists a threshold circuit 
$C$ computing $f$ with $\EC_{\calT}(C) \le cost(T)+1$ (Lemma 5 of~\cite{UDM06}).

Since $T$ has at most $2^{\DT_{th}(f)}$ many leaves, applying (1), we get a 
threshold decision tree $T'$ with $cost(T') \le \log (\#\text{ leaves of }T) \le 
\DT_{th}(f)$.  The result now follows by applying (2) to $T'$.
\end{proof}

\noindent
\textbf{Energy Complexity and Formulas:}
For a Boolean formula $F$, define $\fEC(F)$ is the worst case energy complexity of 
the formula $F$ over the Boolean basis $\calB$. We define,  $\fEC(f)$ as $\min_F 
\fEC(F)$ where $F$ is formula (over the Boolean basis $\calB$) computing $f$. See~
\cref{sec:formula} for more details.

In the rest of the sections, by \textit{circuits}, we refer to Boolean circuits over the Boolean basis $\calB$. 

\section{Energy Complexity and Decision Trees}\label{sec:energy:dt}
In this section, we show a new technique to obtain upper bounds 
on $\EC(f)$.

Recall that any $n$ bit function $f$ can be computed by a circuit 
of energy at most $3n-1$ (\cref{ub:fanin-bounded}). In this section, we identify 
the property of having low depth decision trees as a sufficient condition to 
guarantee energy efficient circuits.
More precisely, we show that for any Boolean function $f$, $\EC(f) \le O(\DT(f)^3)
$. 

One of the challenges in constructing a Boolean circuit is to use as few negation 
gates as possible. The reason is that non-leaf negation gates always contribute to 
the energy since either the 
gate or its input will always output a $1$ on any input to the circuit. We achieve 
this in our construction via an idea inspired by the \emph{connector circuit} 
introduced by Markov~\cite{Mar58}. Before describing the construction, we need the 
following result (\cref{markov:spl}) which helps in controlling the number of 
negation gates in our construction. 
\begin{lemma}\label{markov:spl}
Let $f_0$ and $f_1$ be any two Boolean functions on $n$ variables computed by 
Boolean circuits $C_0$ and $C_1$  respectively. Fix an $i \in [n]$. Define $f(x) = 
(\neg x_i \land f_0(x) ) \lor (x_i \land f_1(x))  $. Then, a circuit $C$ computing 
$f$  
can be obtained using $C_0$ and $C_1$ such that $\negs(C) = 1+\max\set{\negs(C_0), 
\negs(C_1)}$.\end{lemma}
Note that the existence of the circuit in~\cref{markov:spl} can also  be argued using the 
result of Markov~\cite{Mar58} (see Section 10.2 of Jukna~\cite{juk12}) for an arbitrary $f$. However, the construction obtained by directly using the result of Markov can  potentially have high 
energy and hence is not suitable in our context. Since the Boolean function $f$ which we intent to compute is structured, 
we take advantage of 
this observation to adapt Markov's construction and obtain a low energy circuit (with minimal number of negation gates) in~\cref{markov:spl} which is then used 
to prove the 
main result of this section~(\cref{energy:dt}).
\begin{proof}[Proof of~\cref{markov:spl}]

 We start with the circuit $A=(\neg x_i \land 
C_0(x)) \lor (x_i \land C_1(x))$ which uses $1+\negs(C_0)+\negs(C_1)$ negations to 
compute $f$. If one of $\negs(C_0)$ or $\negs(C_1)$ is zero, then $A$ is the 
required circuit. Otherwise, we modify this circuit in $\min\set{\negs(C_0), 
\negs(C_1)}$ 
steps where, in each step, we reduce the number of negations by $1$ such that the 
resulting circuit computes $f$ correctly. 
Hence the resulting circuit $C$ has $1+\negs(C_0)+\negs(C_1) - \min\set{\negs(C_0), 
\negs(C_1)} = 1+\max\set{\negs(C_0), \negs(C_1)} $ negations.

We describe the modifications starting with $A$. Let $g_0$ be a gate in $C_0$ that 
feeds into a negation gate such that the function computed at $g_0$ does not depend 
on the output of any negation gate. Let $D_0$ be the sub-circuit rooted at $g_0$. 
Similarly, let $g_1$ be a gate in $C_1$ with the similar property and let $D_1$ be 
the sub-circuit rooted at $g_1$.
We remove the negation gates that $g_0$ and $g_1$ feeds into from $C_0$ and $C_1$ 
respectively and construct the \emph{connector circuit} (as shown in the box in~
\cref{markov:spl:diag}). We feed the output from the selector as output of the 
negation gates in $C_0$ and $C_1$. Let $D_0'$ (resp. $D_1'$) be the circuit $C_0$ 
(resp. $C_1$) with the output of selector circuit acting as output of the negation 
gate with the negation gate alone removed (Note that we do not completely 
disconnect the sub-circuits from the circuit. The wires connecting $D_0$ (resp. 
$D_1$) to $D_0'$ (resp. $D_1'$) are not shown in~\cref{markov:spl:diag} to avoid 
clutter).

When $x_i=0$, we claim that this circuit outputs $C_0(x)$. This is because when 
$x_i=0$, $D_0'$ gets $\neg D_0$ as output of $g_0$ correctly and hence computes 
$C_0(x)$ while the output of $D_1'$ is inhibited by the $\land$ gate which it feeds 
into. By a similar argument, this circuit computes $C_1(x)$ when $x_i=1$. Hence the 
resulting circuit indeed computes $f$ correctly.

Observe that the number of negation goes down by one in each step since we replace 
two negations by one. We repeat the previous steps restricted to gates in $D_0'$ 
and $D_1'$ as long as the negations in at least one of the circuits is exhausted. 
By the earlier argument, the final circuit $C$ correctly computes $f$.
\end{proof}
\begin{figure}[htp!]
    \begin{subfigure}[t]{.49\textwidth}
	\centering	
	\includegraphics[width=0.65\linewidth]{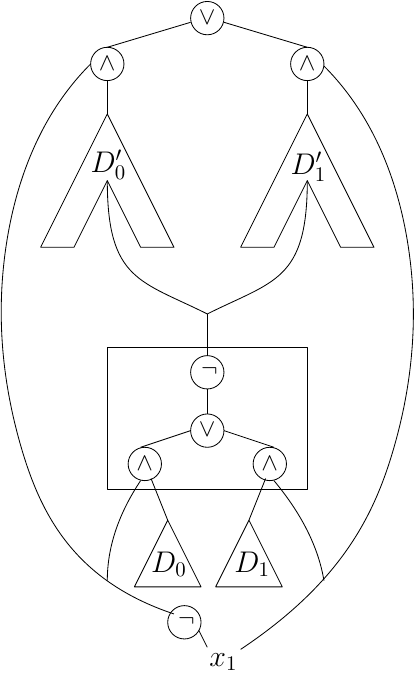}
    \caption{First step in the construction of $C$ in the proof of ~\cref{markov:spl}.}
    \label{markov:spl:diag}
    \end{subfigure}\hfill
    \begin{subfigure}[t]{.49\textwidth}
	\centering
	\includegraphics[width=0.80\linewidth]{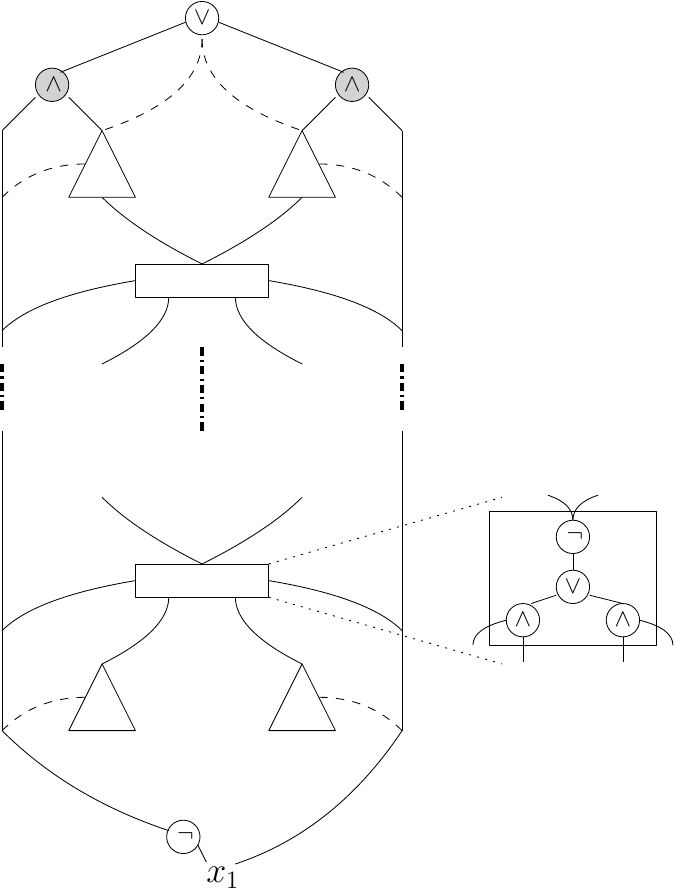}
 	\caption{Circuit $C'$ in the proof of~\cref{circuit:energy:dt}}
	\label{ckt:markov-mod}
	\end{subfigure}
	\caption{Energy efficient circuit construction}
\end{figure}

When $x_1=0$ (resp. $x_1=1$) the part of the circuit computing $f_1$ (resp. $f_0$) 
is not necessary in computing $f$. 
Having obtained a circuit construction which minimizes the usage of negations, we 
need a way to ``turn off'' such gates  that are not needed in computing $f$. In~
\cref{circuit:energy:dt}, we demonstrate how this is achieved, thereby saving 
energy, at the cost of increasing the fan-in of $\land$ gates.

\begin{lemma}\label{circuit:energy:dt}
For any non-constant Boolean function $f$, there exists a circuit $C$ computing $f$ 
with,
(1) all $\lor$ gates are of fan-in $2$ and all $\land$ gates are of fan-in at most 
$\DT(f)+2$, (2) no $\lor$ gate have a negation gate or a variable directly as its 
input, (3) $\negs(C) \le \DT(f)$ and, (4) $\EC(C) \le 2\DT(f)^2$.
\end{lemma}

\begin{proof}
We describe the construction of the circuit by an induction on $\DT(f)$. 

\noindent
\textsf{Base Case: }For $f$ with $\DT(f) =1$, $f$ is either a variable or its 
negation and hence the 
trivial circuit satisfies (1) to (4). For $\DT(f) =2$, let $T$ be an optimal 
decision tree with $x_1$ as its root. Then, $f$ can be computed by the 
circuit $C = (\neg x_1 \land \ell_1) \lor (x_1 \land \ell_2)$ where $\ell_1, 
\ell_2$ could be a variable, negation of a variable or a constant. Also, if $C$ has 
$3$ negations, we apply~\cref{markov:spl} to get a circuit with two negations.
Hence, condition (3) is satisfied. In either cases, the conditions (1) and (2) are 
also satisfied and it can be verified that the energy of the circuit is at most $5 
\le 2\DT(f)^2$. Hence, condition (4) is also satisfied and the base case holds.

\noindent
\textsf{Inductive Step: }Let $f$ be a Boolean function computed by a decision tree 
$T$ of depth $\DT(f)\ge 3$. By induction, assume that for any Boolean function $g$ 
with $\DT(g) \le \DT(f)-1$ there exists a circuit $C$  computing $g$ satisfying (1) 
to (4).
Let the root variable of $T$ be $x_1$ and $T_0$ (resp. $T_1$) be the left (resp. 
right) subtree computing the function $f_0 = f|_{x_1=0}$ (resp. $f_1 = f|_{x_1=1}
$). Since $f_0$ 
and $f_1$ are computed by decision trees of depth $\DT(f)-1$, by induction, there 
exists 
circuits $C_0$ and $C_1$ computing $f_0$ and $f_1$, respectively, satisfying  (1) 
to (4).

Observe that $f(x) =  (\neg x_1 \land f_0) \lor (x_1 \land f_1) $. Hence by~
\cref{markov:spl}, there exists a circuit $C$ computing $f$ (\cref{ckt:markov-mod} omitting the thinly dashed lines) with $\negs(C) = 
\max\set{\negs(C_0) ,\negs(C_1)}+1$. We modify the circuit $C$ as follows: for each 
$\land$ gate which was originally in $C_0$ (resp. $C_1$), we add $\neg x_1$ (resp. $x_1$) 
as input thereby increasing its fan-in by $1$. We also remove the $\land$ gate (shaded in 
~\cref{ckt:markov-mod}) feeding into the top $\lor$ gate and feed the output of the circuits directly to 
the top $\lor$ gate (shown as dashed in~\cref{ckt:markov-mod}). Call the resulting circuit $C'$ and the 
gates from $C_0$ as $C_0'$ (the left part in~\cref{ckt:markov-mod}) and the gates from $C_1$ as $C_1'$ (the right part in~\cref{ckt:markov-mod}).

We first argue that the conditions (1) and (2) holds true for $C'$. We then argue that 
$C'$ correctly computes $f$ using which we argue (3) and (4)   thereby completing 
the induction.

We observe that the condition (1) holds since $\lor$ gate has fan-in $2$ by construction and 
$\land$ gate has fan-in at most $\max\set{\DT(f_0)+3,\DT(f_1)+3}$ which is at most $
\DT(f)+2$. The removal of the shaded $\land$ gates never causes a variable or a negation to be fed to the top $\lor$ gate since $f_0$ and $f_1$ have a decision tree depth of at least $2$ and hence the circuits of the respective functions have top gate as $\lor$ which is guaranteed by base case for depth $2$ and by induction otherwise. Hence, the condition (2) holds. We now argue that $C'$ correctly computes $f$. When $x_1=1$, all the $\land$ gates in $C_0'$ 
evaluate to $0$. Since no input variable or negation gate feeds into any $\lor$ gate in 
$C_0'$ (condition (2)), all the $\lor$ gates and $\land$ gates output $0$ irrespective of 
the remaining input bits. Hence, $C_0'$ outputs $0$. Since $x_1=1$, $C_1'$ behaves 
exactly same as $C_1$. By~\cref{markov:spl}, the circuit $C_1$ correctly computes $f$ 
when $x_1=1$. Hence, the circuit $C'$ correctly computes $f$ for $x_1=1$. The same 
argument with $C_0$ and $C_1$ interchanged can be used to show that $C'$ correctly 
computes $f$ with $x_1=0$.

Since no new negations were added in $C'$, $\negs(C') = \negs(C)$ which, by~
\cref{markov:spl}, equals $\max\set{\negs(C_0), \negs(C_1)} + 1 \le \max\set{\DT(f_1), 
\DT(f_2)}+1 \le \DT(f)$. Hence the condition (3) holds. We show that condition (4) holds 
for $C'$. Let $x$ be an input with $x_1=1$. We have already argued that when $x_1=1$, none of the $\land$ or $\lor$ gates of $C_0'$ output a $1$. Hence the gates that can output a $1$ in $C'$ are the 
negations in $C_0'$, the gates that output $1$ in $C_1'$, the selector gates (in the 
construction of~\cref{markov:spl}), the root gate and the negation gate for $x_1$ (recall that the shaded $\land$ gates are removed). 
Observe that the negations in $C_0'$ is at most $\negs(C_0)$ and 
$C_1'$ behaves exactly as $C_1$ for $x_1=1$. Also, the number of selector circuits used 
in~\cref{markov:spl} is at most $\max\set{\negs(C_0), \negs(C_1)}$.\footnote{While this quantity should be the minimum of the negations of $C_1$ and $C_2$, as seen in the proof of~\cref{markov:spl}, we upper bound this by the maximum of negations.} In addition, each such circuit 
can have at most $2$ gates that output $1$ on any input (see~\cref{ckt:markov-mod}). Hence, 
$\EC(C',x) \le \alpha_0 = \negs(C_0) + \EC(C_1) + 2\max\set{\negs(C_0), \negs(C_1)} + 2$. For $x$ with $x_1=0$, by 
a similar argument, $\EC(C',x) \le \alpha_1 = \negs(C_1) + \EC(C_0) + 2\max\set{\negs(C_0), 
\negs(C_1)} + 2$. Hence, 
$\EC(C')  \le \max\set{\alpha_0,\alpha_1}$ which is at most $ \max\set{\EC(C_0) , \EC(C_1)} +  3\max\set{\negs(C_0) , \negs(C_1)} + 2$. By induction, we have $\EC(C') \le  2(\DT(f)-1)^2+3(\DT(f)-1)+2$ which implies $\EC(C') \le 2\DT(f)^2$ as $f$ is non-constant. This completes the induction.
\end{proof}

We prove the main result of this section.
\begin{reptheorem}{energy:dt}
For any Boolean function $f$, $\EC(f) \le O(\DT(f)^3)$.
\end{reptheorem}
\begin{proof}
If $f$ is constant, the result holds. Otherwise, applying~\cref{circuit:energy:dt}, we have a circuit $C'$ computing $f$ with fan-in of $\lor$ gate being $2$ and 
fan-in of $\land$ gate being at most $\DT(f)+2$ of energy at most $2\DT(f)^2$. Without loss of generality, let $x_1$ be the variable at the root of the decision tree. By construction, all the unbounded fan-in $\land$ gates of the circuit $C'$ have $x_1$ or $\neg x_1$ as an input.

To obtain a bounded fan-in circuit from $C'$, we replace each of the $\land$ gates by a fan-in $2$ circuit as follows. 
Let $g$ be a $\land$ gate of the circuit $C'$ of fan-in $c \le \DT(f)+2$ which takes in $\ell \in \set{x_1, \neg x_1}$ as one of its input. We replace $g$ 
by a tree of fan-in $2$ $\land$ gates of $c$ leaves and of depth $c-1$ with $\ell$ as a leaf at depth $c-1$ as shown in~\cref{fig:fanin}.

\begin{figure}[htp!]
    \begin{subfigure}[t]{.49\textwidth}
	\centering	
	\begin{tikzpicture}[scale=0.8, every node/.style={transform shape}]
	\node[draw,circle] (n1) {$\land$};
	\node[below left=0.5cm of n1] (v1) {};
	\node[below =0.5cm of n1, xshift=-0.25cm] (v2) {};
	\node[below =0.5cm of n1, xshift=0.25cm] (v3) {$\ell$};
	\node[below right=0.5cm of n1] (v4) {};

	\path[draw] (n1) -- (v1);
	\path[draw] (n1) -- (v2);
	\path[draw] (n1) -- (v3);
	\path[draw] (n1) -- (v4);
	\end{tikzpicture}
    \caption{$\land$ gate of fan-in $c$ }
    \end{subfigure}\hfill
    \begin{subfigure}[t]{.49\textwidth}
	\centering
	\begin{tikzpicture}[scale=0.8, every node/.style={transform shape}]
	\node[draw,circle] (n1) {$\land$};
	\node[below right of=n1] (rn1) {};
	\node[draw,circle,below left of=n1] (n2) {$\land$};
	\node[below right of=n2] (rn2) {};
	\node[draw,circle,below left of=n2] (n3) {$\land$};
	\node[below right of=n3] (rn3) {};
	\node[below left of=n3] (n4) {$\ell$};
	\path[draw] (n1) -- (n2) -- (n3) -- (n4);
	\path[draw] (n1) -- (rn1);
	\path[draw] (n2) -- (rn2);
	\path[draw] (n3) -- (rn3);
	\end{tikzpicture}
 	\caption{Tree of depth $c-1$}
	\end{subfigure}
	\caption{Handling $\land$ gates of large fan-in $c$ (for $c=4$)}
	\label{fig:fanin}
\end{figure}

We now argue that this replacement with $\ell$ reattached as the leftmost leaf can only increase the overall energy by a factor of at most $c-1$. 
Consider an input for which $\ell=0$. Then, irrespective of the values of other $c-1$ inputs, none of the fan-in $2$ $\land$ gates output a $1$ as $\ell$ forces all $\land$ gates to output $0$. On the other hand if $\ell=1$, then the added gates can contribute an energy of at most $c-1$. Hence, for any input, the $\land$ gates in $C'$ that output a $0$ does not contribute any energy and those that output a $1$ can contribute of an energy of at most $c-1 \le \DT(F)+1$. Since, in the worst case, all the gates that output a $1$  can be an $\land$, $$\EC(f) \le \EC(C) \cdot (\DT(f)+1) \le 2\DT(f)^2 \cdot (\DT(f)+1) = O(\DT(f)^3)$$ 
\end{proof}

\section{Lower Bounds on Energy Complexity}
In this section, we introduce new methods to show lower bounds on energy 
complexity of Boolean functions. We introduce the notion of continuous 
positive paths (\cref{sec:psens}) using which we prove two energy lower 
bounds. Firstly, we show that the positive sensitivity of a function is a 
lower bound on its energy complexity (\cref{psens:lb}). Secondly, we show 
that for monotone Boolean functions, the cost of the monotone Karchmer-
Wigderson game for the function is a lower bound on its energy complexity 
(\cref{sec:kw:lb}). We conclude the section by proving an  energy lower bound 
of $\Omega(n)$ for any depth $3$ unbounded fan-in circuit computing parity 
function on $n$ bits (\cref{sec:depththree}).

\subsection{Energy Lower Bounds from Positive Sensitivity}\label{sec:psens:lb}
In this section, we prove \cref{psens:energy} from the introduction. 
We first describe an outline here. As a starting case, consider a monotone circuit $C$ computing $f$ evaluates to $1$ on an input $a\in 
\zon$. Let $i \in [n]$ 
be such that $a_i=1$ and flipping $a_i$ to $0$ causes the circuit to evaluate to $0$. We show that for such an index $i$
on input $a$, there is a path from $x_i$ to the root such that all the gates in the path outputs a $1$. The latter already implies a weak energy lower bound. We then generalize this idea to non-monotone circuits as well and use it to prove energy lower bounds. This generalization also helps us to prove upper bounds for $\KW^{+}$ games in~\cref{sec:kw:lb}.

To keep track of all input indices that are sensitive in the above sense, we introduce the measure of positive sensitivity denoted by $\psens(f)$ (as defined in~\cref{sec:prelims}). 
For example, the functions $f \in \set{\oplus_n, \land_n}$ have $\psens(f) = n$ while $\psens(\lor_n) = 1$. Let $\psenss(f,a)$ denote the set of positive 
sensitive indices on $a$.
 
\subsubsection{Continuous Positive Paths}\label{sec:psens}

Let $C$ be a Boolean circuit computing $f:\zon \to \zo$. For an input $a \in \zon$, we call a path of gates such that every gate in the path output $1$ on $a$ as a \emph{continuous positive path} in $C$.

 Fix an 
$a \in \zon$. We argue that for  every positive sensitive 
index $i$ on $a$, either there is a continuous positive 
path from $x_i$ to the root or it must be broken by a 
negation gate of the circuit. Using this we show, in 
the next section, that energy complexity of a function is 
lower bounded by its positive sensitivity.
\begin{lemma}\label{cpath:gen-circuit}
Let $f:\zon \to \zo$ and $a \in \zon$ be an input such 
that $\psens(f,a) \ne 0$ and $i \in \psenss(f,a)$. Let $C$ be any 
circuit computing $f$. Then,  either (1) there is a continuous positive path from $x_i$ to root or (2) $x_i$ directly feeds into a negation gate or (3) there 
is a continuous positive path from $x_i$ to a gate which feeds into a negation gate 
of $C$.
\end{lemma}
\begin{proof}
It suffices to prove the following stronger statement: for a Boolean function $f$ and an $a \in \zon$ with $\psens(f,a) \ne 0$ and $i \in \psenss(f,a)$, let $C$ be any circuit such that $C(a) = f(a)$ and $C(a\oplus e_i) = f(a \oplus e_i)$. Then, either (1) there is a continuous positive path from $x_i$ to root or (2) $x_i$ directly feeds into a negation gate or (3) there is a continuous positive path from $x_i$ to a gate which feeds into a negation gate of $C$. 

Proof is by induction on $\negs(C)$. Let $C$ be any circuit such that $f(a) = C(a)$ and $f(a \oplus e_i) = C(a \oplus e_i)$. 

\noindent {\sf Base Case:} For the base case, $\negs(C) = 0$ and $C$ is a monotone circuit. By definition, $i \in \psens(f,a)$ implies that $a > a\oplus e_i$. We claim that if $i \in \psens(f,a)$, then $C(a) = 1$. For a contradiction, suppose that $f(a) = C(a) =  0$. Then $C(a \oplus e_i) = 0$ because $C$ is monotone. But then $f(a \oplus e_i) = 0$ which contradicts the fact that $i \in \psens(f,a)$.

Since $C(a)$ outputs $1$ and since $C$ is a monotone circuit, the root being an 
$\lor$ or $\land$ gate must have a child gate evaluating to $1$. Since this gate is 
again $\lor$ or $\land$ the same argument applies implying that there exists a series of gates all 
evaluating to $1$ reaching some inputs.
For any $i \in \psens(f,a)$,  we show that there must a path from $x_i$ to the 
root with all the gates in the path evaluating to $1$ in $C$ on input $a$ (implying that (1) holds). 

For a 
contradiction, suppose that every path from $x_i$ to the root gate passes via some 
gate that evaluates to $0$.  Among all the paths from $x_i$ to the root, collect all 
the gates that evaluate to $0$ for the first time in the path and call this set as 
$T$. We fix all the variables except $x_i$ to the values in $a$ and view each of 
the gates in $g \in T$ as a function of $x_i$. Now, flipping $x_i$ from $a_i=1$ to 
$0$ does not change the output of any $g \in T$ as they compute monotone functions 
and already evaluate to $0$. Since all other values are fixed, the output of the root 
gate does not change by this flip which contradicts the fact that $i \in \psens(f,a) $.

\noindent{\sf Induction Step:} Let $C$ be a circuit with $f(a) = C(a)$ and $f(a \oplus e_i) = C(a \oplus e_i)$ and $\negs(C) \ge 1$. If $x_i$ feeds directly into a negation gate, then statement (2) holds as required. Otherwise, let $g$ be the first gate that feeds into a negation in the topologically sorted order of the gates of $C$. 

We have the following two possibilities. In both the cases, we argue existence of 
continuous positive path in $C$ from the variable $x_i$, thereby completing the induction. 
\begin{itemize}
\item {\it On 
input $a$, flipping $a_i$ change the output of $g$}. Denote the function 
computed at $g$ as $f_g$. Then $f_g$ is monotone and $i \in \psenss(f_g,a)$ and is 
non-empty. Hence, applying the argument in the base case to $f_g$ and the monotone 
circuit rooted at $g$, we are guaranteed to get a continuous positive path from 
$x_i$ to $g$. Since the circuit at $g$ is a sub-circuit of $C$ (that is, it appears 
as an induced subgraph), this gives a continuous positive path in $C$ also. 
\item {\it On input $a$, flipping $a_i$ does not change the output of $g$}.  
In this case, we remove the negation gate that $g$ feeds into and hard-wire the 
output of this negation gate (on input $a$) in $C$ to get a circuit $C'$. 
Note that all other gates in $C$ are left intact. Observe 
that $C'(a) = f(a)$. Since flipping $a_i$ did not change the output of $g$ and as all other gates are left intact, $C'(a \oplus e_i) = f(a\oplus e_i)$. As $\negs(C') = \negs(C)-1$, by induction, either (1) there is a continuous positive path from $x_i$ to root or (2) there is a continuous positive path from $x_i$ to a gate which feeds into a negation gate of $C'$. 
By construction, $C'$ is same as $C$ except for the negation gate. Hence, a continuous positive path in $C'$ is also a continuous positive path in $C$.
\end{itemize}
\end{proof}

\subsubsection{From Positive Sensitivity to Energy Lower Bounds}\label{psens:lb}
We call the negation gates and the root gate of a circuit as \emph{target gates}.
In \cref{cpath:gen-circuit}, we have already shown the existence of continuous positive paths from a positive sensitive index up to a target gate. Using this, we show an energy lower bound for any circuit of bounded fan-in computing a Boolean function $f$ in terms of $\psens(f)$.  Since the fan-in of the circuit is limited, we exploit the idea that in a connected DAG, the number of internal nodes (in-degree at least $1$)  is lower bounded by the number of source nodes (in-degree $0$).

Since every such positive sensitive index is reachable via a continuous positive path from a target gate, we obtain a lower bound on energy by applying this idea on an appropriate subgraph constructed from our circuit.
\begin{reptheorem}{psens:energy}
For any Boolean function $f:\zon \to \zo$ computed by a circuit $C$ over the Boolean basis, $ \EC(C) \ge \psens(f)/3$. 
\end{reptheorem}
\begin{proof}
Without loss of generality assume that $f$ is non-constant. Let $C$ be any circuit computing $f$ of fan-in $2$ such that $\EC(C) = \EC(f)$. We 
prove that $\forall a \in \zon, \psens(f,a) \le 3\EC(C)$.

Let $a \in \zon$ by any input. If $\psens(f,a) = 0$, the claim holds. Hence we can 
assume, $\psens(f,a) \ne 0$. Let $T$ be the set of all target nodes in $C$. 
For every $i \in \psenss(f,a)$, by~\cref{cpath:gen-circuit}, there exists continuous positive paths starting from $x_i$ to a gate $g \in T$.

For every $g \in T$, let $X_g$ be the set of all gates that lie in a continuous 
positive path from an $x_i$ to $g$ for some $i \in \psenss(f,a)$. Note that the subgraph induced by vertices in $X_g$ is connected and does not include $g$.  
We now obtain a connected DAG with $\psens(f,a)$ leaves as follows. Let $D$ be a full binary tree (with edges directed from child to parent) with $|T|$ many leaves and hence $|T|-1$ internal nodes. For each $g \in T$ if it is a negation, we attach the gate feeding into $g$ as a leaf of the $D$ and if it is a root, we attach the root as a leaf of the $D$. We call the resulting DAG as $H$. 

Let $X = \cup_{g \in T} X_g$. We now argue that $\psens(f,a) \le |X|+|T|$. Observe that the number of internal nodes in the DAG $H$ is $|X| + (|T|-1)$ where the first term is the gates in $X$ and the second term is the number of internal nodes of the tree $D$. 
Since graph induced on $X_g$ is connected for each $g$, this results in the DAG $H$ to be connected with $\psens(f,a)$ many source nodes. As $H$ is connected, the number of leaves, which is $\psens(f,a)$, is at most the number of internal nodes $+1$ which is $|X|+|T|-1+1 = |X|+|T|$.

We now give a bound on $|X|$ and $|T|$. Recall that since the target gates include negations and the root, $|T| = \negs(C)+1$. Since $\negs(C)$ in any circuit $C$ is at most $\EC(C)$ (\cref{energy:negs}), $|T| \le \EC(C)+1$. Since all gates in $X$ output $1$ as they belong to some continuous positive path in $C$, $|X| \le \EC(C)$. This implies that, $\psens(f,a) \le |X|+|T| \le 2\EC(C)+1 \le 3\EC(C)$ as $f$ is non-constant.

\end{proof}

This implies that since $\psens(\land_n) = n$, for  $\EC(\land_n) \ge n/3$  which is asymptotically tight by~\cref{ub:fanin-bounded}. We observe that 
even thought $\land_n$ is symmetric, the result of Suzuki \etal~\cite{SUZ13} 
on the energy lower bound on threshold circuits computing symmetric functions (which applies to Boolean circuits too), only yields a trivial lower bound (see~\cref{tab:tradeoff}). 
We remark that both these bounds does not give any non-trivial energy lower bound for $f = \lor_n$.
Note that~\cref{psens:energy} uses the fact that the circuits used have fan-in $2$. If fan-in of the circuit $C$ is $c$, then replacing each gate by a tree of $c-1$ gates of fan-in $2$, by a similar argument as before, $\EC(C) \ge \psens(f)/(c+1)$.

\subsection{Energy Lower Bounds from Karchmer-Wigderson Games}\label{sec:kw:lb}
\cref{energy:mon} says that the monotone circuits are not energy efficient for computing monotone functions. In this section, we explore the limits on how  energy efficient can non-monotone circuits be in computing monotone functions by showing the following. 
\begin{reptheorem}{kw:lb}
Let $f:\zon \to \zo$ be a monotone function. Then $\EC(f) = \Omega(\KW^+(f))$.
\end{reptheorem}

Our approach is to use~\cref{cpath:gen-circuit} and utilize the existence of 
continuous positive paths to design a $\KW^{+}$ protocol of cost $\EC(C)\log (\text{fan-in}(C))$ in \cref{lb:kw:energy} which immediately proves the above theorem. For the perfect matching function $f_{PM}$ on a graph of $n$ edges since $\KW^+(f_{PM}) = \Omega(\sqrt{n})$~\cite{RW92}, this implies that any circuit $C$ with constant fan-in computing $f_{PM}$ require an energy of $\Omega(\sqrt{n})$.

Recall that for $x \in f^{-1}(1)$ and $y \in f^{-1}(0)$, $S^+_f(x,y) \defn \set{ i \mid  x_i =1, y_i =0, i \in [n]}$. Also, we call the set of all negation gates, along with the root gate of $C$ as the \emph{target gates} of $C$. 
\begin{lemma}\label{lb:kw:energy}
For a non-constant monotone Boolean function $f$, let Alice and Bob hold inputs $a \in f^{-1}(1)$ and $b \in f^{-1}(0)$ respectively. Let $C$ be any circuit computing $f$, and every gate in the circuit is either a $\land, \lor$ with fan-in of at most $c$ or a negation gate. Then, $\KW^+(f) \le \EC(C) \log c$.
\end{lemma} 
\begin{proof}
We argue that, without loss of generality it can be assumed that $\psenss(f,a) = \{ i \mid a_i = 1 \}$. 
Alice finds an $a' \prec a$ with $f(a') = f(a) = 1$ such that for any $a'' \prec a'$, $f(a'') = 0$. Observe that $a' \ne 0^n$ for otherwise, $f(0^n) = 1$ and since $f$ is monotone, $f$ must be a constant which is a contradiction.
By construction,  every bit in $a'$ which is 1 is sensitive.
Since $a' \prec a$, $S_f^+(a',b) \sse S_f^+(a,b)$, thereby it suffices to find an index in $S_f^+(a',b)$.

We now describe the protocol. Let $a \in f^{-1}(1)$ such that $\psenss(f,a) = \{ i \mid a_i = 1 \}$. Before the protocol begins, Alice does the following pre-computation.
Let $\calP$ be the collection of positive paths one each for every $i \in \psenss(f,a)$, which exists as per \cref{cpath:gen-circuit}. Alice computes $\calP = \bigcup_{g \in T} \calP_g$ where $\calP_g$ is the collection of all continuous positive paths ending at the target gate $g$. This ends the pre-processing.

Now Alice and Bob fixes an ordering of the target gates. For each target gate $g \in T$ in the order, the following procedure is repeated.
For each continuous positive path $p \in \calP$, ending at $g$, Alice sends the  address of the previous gate in the path $p$ (using $\log c$ bits) until they trace back to an input index $i$. Now, Bob checks if $b_i = 0$, and if so, we have found $i \in S_f^+(a,b)$, else, they attempt on the next $p \in \calP_g$.

We argue about the correctness of the protocol. Notice that the above protocol searches through all $i \in \psenss(f,a)$ by traversing through all $\calP_g$, for $g \in T$. Since $\psenss(f,a) = \{ i \mid a_i = 1 \}$  and $S_f^+(a,b) \sse \psenss(f,a)$ the protocol correctly computes $i$ such that $a_i=1$ and $b_i=0$. We now analyze the cost of the protocol. Observe that whenever Alice and Bob encounters a new path in $\calP$ with parts of the path already traversed, they can move along the edges without exchange of any address. Hence, a gate in a path belonging to $\calP$ gets visited at most once. Since the protocol visits only those gates that output $1$ at most once on the input $a$, we have a protocol with communication cost $ \le \EC(C,a) \times \log (c) \le \EC(C) \log c$.
\end{proof}
\subsection{Energy Lower Bounds for Depth Three Circuits}\label{sec:depththree}
We now consider lower bounds on the energy complexity  of constant depth (unbounded fan-in) circuits computing the parity function on $n$ bits. Observe that the energy complexity of circuits of unbounded fan-in can be very small compared to bounded fan-in circuits\footnote{For instance, $\land_n$ has a fan-in $n$ circuit of depth $1$ and energy $1$ computing it while, by~\cref{psens:energy}, energy of any bounded fan-in circuit computing the same function is $\Omega(n)$.}. Hence the results proved in the bounded fan-in setting (~\cref{psens:energy}) does not directly apply.

For any Boolean function $f$, the trivial depth $2$ circuit of unbounded fan-in computing $f$ has an energy $n+2$ and it can be shown that any depth two Boolean circuit computing the parity on $n$ bits require an energy of $n+1$. 
\begin{proposition}
Let $C$ be any depth $2$ circuit computing $\oplus_n$. Then $\EC(C) \ge n+1$.
\end{proposition}
\begin{proof}
Since $C$ computes $\oplus_n$, no variable or its negation can feed into the root gate,  and every variable or its negation must feed into all gates at depth $2$.

We now argue that at least $n-1$ variables must be negated in $C$. Suppose not, 
then there must be two variables, say $x_i$ and $x_j$, that feeds into all the  gates in depth $2$ unnegated. Setting $x_i=0$, all the $\land$ gates at 
depth $2$ must evaluate to $0$. Similarly, setting $x_j=1$, causes all $\lor$ gates in depth $2$ to evaluate to $1$. Hence the circuit $C$ evaluates to a fixed value irrespective of the remaining $n-2$ inputs unset 
which is a contradiction. Thus we conclude that at least $n-1$ variables must be negated. Consider an input $x$ that is $0$ on these $n-1$ 
negated variables and $1$ on the remaining variable. On this input, all the 
negation gates, the $\land$ gate which they feed into and the root gate 
evaluates to $1$. Hence $\EC(C) \ge \EC(C,x) \ge n-1+1+1 = n+1$.
\end{proof}
Santha and Wilson~\cite{JSW93} showed that for any unbounded fan-in circuit $C$ 
of depth $d$ computing $\oplus_n$, $\negs(C) \ge d( \lceil n / 2 \rceil )^{1/d}-
d$. Since energy complexity of a circuit $C$ is at least the number of negation 
gates in $C$ (\cref{energy:negs}), this implies that $\EC(C) \ge d( \lceil n / 2 
\rceil )^{1/d}-d$ for any such circuit $C$ computing $\oplus_n$.

While we are unable to prove strong lower bounds for circuits of depth $d$ for 
an arbitrary constant $d$, we show that any depth $d=3$  unbounded fan-in 
circuit computing the parity function requires large energy. We achieve this by 
appealing to the known lower bounds on size of any constant depth circuit 
computing $\oplus_n$. 
Razborov showed that any circuit $C$ of depth $d$ of unbounded fan-in computing parity on $n$ bits must be of size at least $2^{\Omega(n^{1/4d})}$~\cite{R87}. Using this result we show an energy lower bound of $\Omega(n)$ for any depth $3$ circuit computing $\oplus_n$.

\begin{reptheorem}{lb:depththree:raz}
Let $C$ be any unbounded fan-in circuit of depth $3$ computing the parity function 
on $n$ variables. Then, $\EC(C)=\Omega(n)$.
\end{reptheorem}

\begin{proof}
We call the root gate of the circuit as the ``top" level and the two 
level immediately below as the ``middle" and ``bottom" 
levels respectively. Note that negations can appear anywhere in the circuit and does not count towards the level.
Assume without loss of generality that the circuit $C$ does not have any redundant gates. 

Let there be $i$ negated input variables and 
without loss of generality assume $i < n$. We set these variables to $0$ and let $C'$ be the resulting circuit obtained. Let $g_1, g_2 ,\ldots, g_k$ be the $k
$ gates in the bottom layer that feed to the layers above via negation gates. We 
set input variables to these $k$ gates such that the output of the negations are 
fixed in the following way: for the gate $g_i$, consider any input variable, say 
$x_j$, that feed into $g_i$ and set it to $0$ if $g_i$ is $\land$ gate and $1$ if 
$g_i$ is $\lor$. We also remove the gates that have become a constant and hardwire 
their output to get the resulting circuit $C'$. Hence, all the gates at the bottom 
level are not fed negated to the level above.

In this process, we have eliminated the $k$ negations leaving us with the circuit 
$C''$ where all the gates at bottom and middle layer computes some monotone function 
on the remaining $m=n-(i+j)$ for some $j \le k$ variables. Since the resulting 
circuit must compute parity on $m$ variables, 
by~\cite{R87}, $size(C'') 
\ge 2^{\Omega(m^{1/12})}$. Since $C''$ is of depth $3$, the number of bottom and 
middle gates in $C'$ must also be at least $2^{\Omega(m^{1/12})}$. 
As the gates in the  bottom and 
middle level computes monotone function, setting all the variables to $1$ in $C''$ 
forces all of them must output $1$ (Here we use 
the fact that the redundant gates are eliminated in $C$). 
Hence in $C$, there is a setting of input such that at least $i+k \ge i+j= n-m$ gates contributes an energy of $1$ (since either the input to the negation or the negation gate itself will be $1$) and $2^{\Omega(m^{1/12})}$ gates in $C$ that evaluate to $1$. Hence, $\EC(C) \ge n-m+2^{\Omega(m^{1/12})}$. Let $c$ be the smallest integer such that for all $t \ge c$, $2^{\Omega(t^{1/12})}$ is at least $t$. There can be two cases: 
\begin{enumerate}
	\item If $m \ge c$, then $\EC(C) \ge n-m+2^{\Omega(m^{1/12})} \ge n$, 
	\item Otherwise, $m < c$ and $\EC(C) \ge n-m+2^{\Omega(m^{1/12})} \ge n-c+1 = \Omega(n)$ as $c$ is independent of $n$.
\end{enumerate}
 Hence $\EC(C) = \Omega(n)$. 
\end{proof}

\section{Energy Complexity of Boolean Formulas}\label{sec:formula}

For a formula $F$, let $\L(F)$ denote the number of leaves in the formula $F$. For any formula $F$, we have $\fEC(F) \le 2\L(F)-1$ as $F$ has $\L(F)-1$ internal nodes and can have at most $\L(F)$ negation gates as leaves (in the worst case). Unlike circuits, any sub-function computed in a formula cannot be reused which can potentially lead to many gates that output a $1$ on some input. For this reason, one would expect that it is unlikely for Boolean formulas to be energy efficient. As a warm up, we first implement the above argument for structured Boolean formulas where we prove strong lower bounds of $\Omega(\L(F))$ (\cref{sec:warmup}) and discuss its limitations. Then, using a different approach, we show a weaker lower bound of $\Omega(\sqrt{\L(F)}-\depth(F))$ (\cref{sec:fml}) for  arbitrary Boolean formulas.

\subsection{A Warm up}\label{sec:warmup} 
We consider the following approach to prove a lower bound on 
energy complexity of a formula $F$ by exhibiting an input on which many 
gates are guaranteed to output a $1$. Suppose $t$ be the number of gates in 
a formula which have both its inputs as variables. We call such gates as 
\textit{non-skew gates}. Now, set the $n$ variables to $0$ or $1$ uniformly 
at random. Then, each of the $t$ gates evaluate to a $1$ with probability 
at least $1/4$. Hence, on expectation, there are at least $t/4$ such gates 
evaluating to $1$. This implies the existence of an input on which $
\Omega(t)$ gates fire which gives the following proposition. 
\begin{proposition}\label{formula:lb:nonskew}
For a formula $F$, let $t$ be the number of non-skew gates in $F$. If $t = \Omega(\L(F))$, then $\fEC(F)=\Theta(\L(F))$.
\end{proposition}
However, this argument fails\footnote{In the conference version of this paper~\cite{DOS18}, it was erroneously claimed that \cref{formula:lb:nonskew} holds for \textit{all} Boolean formulas (that is, irrespective of $t$).} for formulas where  the gates are \textit{skew} 
(\ie~exactly one of the input to the gate is a variable) since
randomly setting the input does not necessarily guarantee a constant 
probability for the skew part to output a $1$ (for example, consider formulas whose underlying graph is the fully right skewed tree). Hence, this approach does not 
give a lower bound for $\fEC(F)$ for an arbitrary formula $F$.

Nevertheless, there can be special formulas for which we can prove the lower bound of $\Omega(\L(F))$. For instance, consider the read-once formulas with negations at leaf. Similar to the argument of energy of monotone circuits (\cref{energy:mon}), the following can be concluded about them (irrespective of the formula structure).
\begin{proposition}
For any read-once formula $F$ with negations at the leaf, $\fEC(F) \ge \L(F)-1$.
\end{proposition}

\subsection{Bounds on Energy Complexity for Boolean Formulas}\label{sec:fml}
In this section, we take a different approach and show the following lower bound on the energy 
complexity of any Boolean formula.
\begin{reptheorem}{energy:formula}
For a Boolean function $f$, computed by a formula $F$, $$ \fEC(F) = \Omega 
\left(\sqrt{\L(F)}-\depth(F)\right ).$$
\end{reptheorem}
Though the above result applies for any Boolean formula, it does not give any non-trivial lower bound for formulas that have large depths due to presence of long path of skew gates.

We now describe our approach. The main idea is to use a structural decomposition result for Boolean 
formulas due to Guo and Komargodski~\cite{GK17} (see also Tal~\cite{T14}). 
They showed that any formula $F$ can be transformed to another ``structured 
formula'' $F'$ without blowing up the size.  More precisely,

\begin{theorem}[Theorem 3.1 of~\cite{GK17}]\label{thm:gk}
Let a Boolean function $f$ be computed by a Boolean formula $F$ with $
\negs(F) \ge 1$. Then, there exists $T \le 5\negs(F)-2$ monotone functions 
$g_1 \ldots, g_T $ where each $g_i$ is computed by a monotone formula 
$G_i$ and a function $h:\zo^T \to \zo$ computed by a read-once formula $H$ 
such that $f(x) = h(g_1,\ldots, g_T)$ computed by the formula $F'\defn 
H(G_1,\ldots,G_T)$ satisfy $\L(F') \le 2\L(F)$.
\end{theorem}

Before proceeding, we illustrate this result for a simple case (which we use later) when the formula $F$ has exactly one negation gate which it not at the root.

\begin{figure}[htp!]
	\begin{subfigure}[t]{.49\textwidth}
		\centering
\includegraphics[scale=0.4]{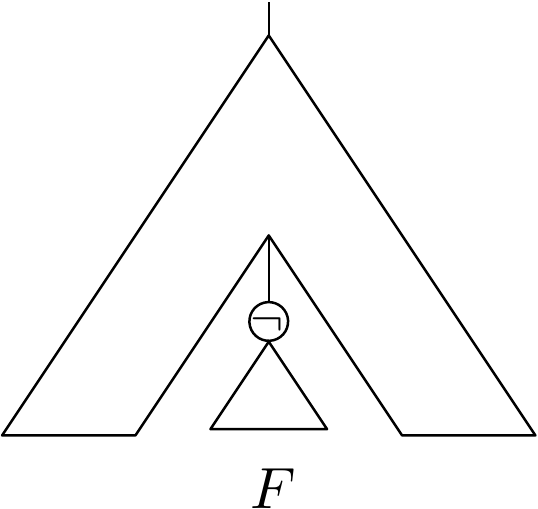}
\caption{Formula $F$}
	\end{subfigure} 
	\begin{subfigure}[t]{.59\textwidth}
		\centering
		\includegraphics[scale=0.4]{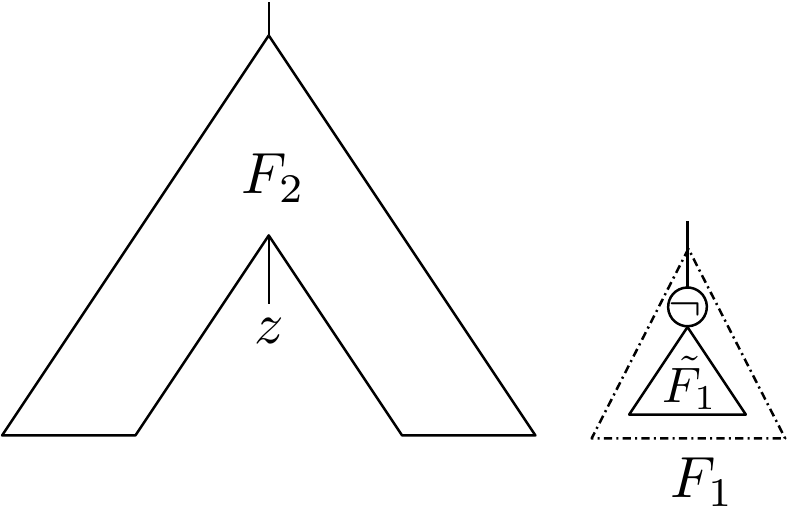}
		\caption{Sub formulas $F_1$ and $F_2$ of $F$}
	\end{subfigure}
	\caption{Theorem 3.1 of~\cite{GK17} for formula $F$ with $\negs(F)=1$}
	\label{fig:split}
\end{figure}

\cref{fig:split}(a) shows the formula $F$ with exactly one negation gate and \cref{fig:split}(b) shows two sub-formulas $F_1$ and $F_2$ with $F_2$ and $\tilde{F_1}$ being monotone.

\begin{figure}[htp!]
\centering
	\includegraphics[scale=0.4]{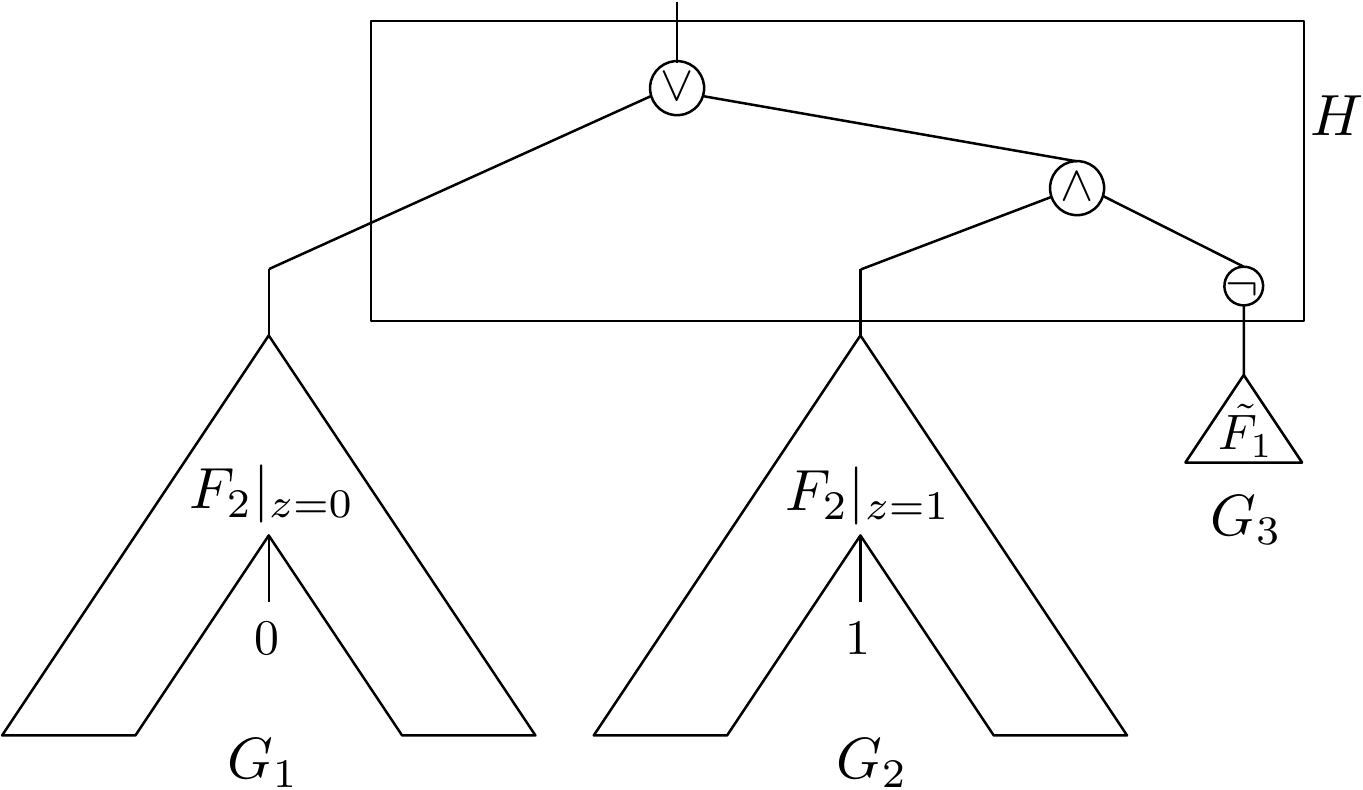}
	\caption{Structured formula $F'$ obtained when Theorem 3.1 of~\cite{GK17} is applied to $F$ in~\cref{fig:split}}
	\label{fig:split:final}
\end{figure}

The resulting formula obtained upon  applying~\cref{thm:gk} is $F' = F_2|_{z=0} \lor (F_2|_{z=1}\land \tilde{F_1})$ as shown in~\cref{fig:split:final}. The monotone formulas $G_1$, $G_2$ and $G_3$ are $F_2|_{z=0}$, $F_2|_{z=1}$ and $\tilde{F_1}$ respectively. The read-once formula $H$ is $H(y_1,y_2,y_3) = y_1 \lor (y_2 \land \neg y_3)$. Also, $\L(F') \le 2\L(F)$.

\noindent
We now describe our proof strategy: firstly, we analyze the energy of the formula $F'$ obtained in~\cref{thm:gk} and 
show (in~\cref{fml:ub}) that $\fEC(F')$ is upper bounded asymptotically by 
$O(\negs(F) \times (\fEC(F)+\depth(F))$. This implies that the decomposition 
in~\cref{thm:gk} is not only size efficient but also energy efficient. 
The specific structure of the formula from \cref{thm:gk} implies that $\fEC(F')$ is lower bounded by $\Omega(\L(F)-\negs(F))$ (\cref{fml:gf:lb}). 
Finally, comparing the upper and lower bound for $\fEC(F')$ gives a lower bound 
on $\fEC(F)$ in terms of $\L(F), \depth(F)$ and $\negs(F)$ using which we 
prove~\cref{energy:formula}. Before proceeding, we need the following 
observation.

\begin{proposition}\label{fml:lb}
Let $F$ be any formula and $g$ be any gate of $F$ other than the root. Let 
$D$ be a formula obtained by replacing the subtree at gate $g$ by a 
variable $z$. Then for any $b \in \zo$, $\EC(D|_{z=b}) \le \fEC(F)+\depth(F)
$.
\end{proposition}
\begin{proof}
Fix a $b \in \zo$ and let $a$ be an input on which $D|_{z=b}$ achieves the maximum energy. Consider the evaluation of gates in $F$ on this input $a$. If we ignore the gates in the subtree rooted at $g$ in $F$, as $F$ is a formula, the evaluation of gates on the input $a$ for $F$ and $D|_{z=b}$ can differ only on those gates that lie in the path from $g$ to the root. Hence, 
$$ \fEC(F) \ge \fEC(F,a) \ge \EC(D|_{z=b},a)-\depth(F) = \EC(D|_{z=b}) - \depth(F)$$
which completes the proof.
\end{proof}

\begin{lemma}[Upper Bound for $\fEC(F')$] \label{fml:ub}
Let $f:\zon \to \zo$ be computed by a Boolean formula $F$ with $\negs(F)\ge 1$. Then, the formula $F'$ computing $f$ obtained by applying the decomposition of~\cref{thm:gk} to $F$ satisfies,
\begin{equation}\label{eq:energy}
\fEC(F') \le (5\negs(F)-2)(\fEC(F)+\depth(F)+1).
\end{equation}

\end{lemma}
\begin{proof}
We proceed by tracing the construction in Theorem 3.1 of~\cite{GK17} (\cref{thm:gk}) where we bound the energy of the resulting formula thereby proving the result.

\noindent
[By strong induction on $\negs(F)$] For the base case with $\negs(F)=1$, let $F_1$ be the minimal formula containing all negations of $F$. If $F_1 = F$, then the root gate of $F$ must be a NOT gate and $F'=F$ satisfies~\cref{eq:energy}. Otherwise, let $F_2$ be the formula obtained by replacing $F_1$ in $F$ by a new variable $z$. As $F_1$ has the only negation gate of $F$, $F_2$ is monotone implying $F_2 = F_2|_{z=0} \lor (F_2|_{z=1} \land z)$. Also there exists a formula $\mytilde{F_1}$ such that $F_1 = \neg \mytilde{F_1}$ (see~\cref{fig:split}). Now the formula $F' = F_2|_{z=0} \lor (F_2|_{z=1} \land \neg \mytilde{F_1})$ computes the same function as $F$. Since $\fEC(F')$ is upper bounded by the energy of the individual formulas and the  connecting gates,
\begin{align*}
\fEC(F') & \le \fEC(F_2|_{z=0}) + \fEC(F_2|_{z=1}) + \EC(\mytilde{F_1}) + 3 \\
& \le 2\fEC(F)+2\depth(F)+\EC(\mytilde{F_1})+3 && [\text{\cref{fml:lb}}]\\
& \le 3(\fEC(F)+\depth(F)+1) && [\text{$\mytilde{F_1}$ is a subformula of $F$}]
\end{align*}

For the inductive case, let $F$ be any Boolean formula with $t=\negs(F)>1$ and the result holds for all formulas with negations less than $t$. Let $F_1$ be the smallest subformula of $F$ that contains all the negations of $F$. There can be two cases.

\begin{description}
\item [Case 1. $F_1$ is same as $F$:] In this case, we show that there is an $F'$ computing the same function as $F$ with $\fEC(F') \le (5\negs(F)-4)(\fEC(F)+\depth(F)+1)$ satisfying~\cref{eq:energy}. Based on the root gate of $F$, there can be two subcases. 

Suppose the root of $F$ is a NOT gate. Then, there exists a formula $E$ such that $F=\neg E$. Since $\negs(E) = \negs(F)-1$, by induction, there exists an $E'$ computing the same function as $E$ with $\fEC(E') \le (5\negs(E)-2)(\fEC(E)+\depth(E)+1)$. Now the formula $F'=\neg E'$ computes the same function as $F$. Estimating $\fEC(F')$, we have
\begin{align*}
\fEC(F') & \le \fEC(E')+1 \\
& \le (5\negs(E)-2)(\fEC(E)+\depth(E)+1) +1&& [\text{Induction}] \\
& \le (5(\negs(F)-1)-2)(\fEC(E)+\depth(F))+1 && [\text{$\depth(E) = \depth(F)-1$}] \\ 
& \le (5\negs(F)-4)(\fEC(F)+\depth(F)+1) && [\text{$\fEC(E) \le \fEC(F)+1$}]
\end{align*}

Suppose the root of $F$ is AND/OR. Without loss of generality, let the root be OR gate. A similar argument holds for the case of AND gate. Then, let $F = E_\ell \lor E_r$. where $E_\ell, E_r$ are the left and right subtrees of the root, respectively. Since $E_\ell$ and $E_r$ are subformulas of $F$, observe that $\fEC(E_\ell) \le \fEC(F)$ and $\fEC(E_r) \le \fEC(F)$. Since $F_1=F$, it must be that  $\negs(E_\ell) \ge 1$ and $\negs(E_r) \ge 1$. Hence, by induction, there exists formulas $E_{\ell}'$ and $E_r'$ computing the same function as $E_\ell$ and $E_r$, respectively. Consider the formula $F' = E_{\ell}' \lor E_r'$. We now show that $F'$ satisfies the required energy bound.
\begin{align*}
\fEC(F') &  \le \fEC(E_{\ell}') + \fEC(E_r')+1  \\
& \le (5\negs(E_\ell)-2)(\fEC(E_\ell)+\depth(E_\ell)+1) \\
 & \quad + (5\negs(E_r)-2)(\fEC(E_r)+\depth(E_r)+1) + 1 && [\text{Induction}] \\
&  \le (5\negs(E_\ell)-2)(\fEC(F)+\depth(F)) \\
 & \quad + (5\negs(E_r)-2)(\fEC(F)+\depth(F)) + 1 && [\text{$\depth(E_\ell),\depth(E_r) \le \depth(F)-1$}] \\
 & \le (5\negs(F)-4)(\fEC(F)+\depth(F)+1)
\end{align*}

\item [Case 2. $F_1$ is not same as $F$:] Let $F_2$ be the formula obtained 
by replacing $F_1$ in $F$ by a new variable $z$. Similar to the argument in 
the base case, $F'=F_2|_{z=0} \lor (F_2|_{z=1} \land F_1)$ computes the same 
function as $F$. Since $F_1$ does not have a smaller subformula containing all  
its negations, we can apply Case 1 to $F_1$ to get a formula $F_1'$ computing 
same function as $F_1$ with $\fEC(F_1') \le (5\negs(F_1)-4)(\fEC(F_1)+
\depth(F_1)+1)$. Hence, 
\begin{align*}
\fEC(F') & \le \fEC(F_2|_{z=0}) + \fEC(F_2|_{z=1}) + \fEC(F_1') + 2 \\
& \le 2\fEC(F) + 2\depth(F) + \fEC(F_1') + 2 && [\text{\cref{fml:lb}}]\\
& \le (5\negs(F_1)-4)(\fEC(F_1)+\depth(F_1)+1) \\
& \quad + 2(\fEC(F)+\depth(F)+1)\\
& \le (5\negs(F)-2)(\fEC(F)+\depth(F)+1) && [\text{$F_1$ is a subformula of $F$}] 
\end{align*}

\end{description}
\end{proof}

\begin{lemma}[Lower Bound for $\fEC(F')$] \label{fml:gf:lb}
Let $F$ be a formula and $F'$ be the formula obtained by applying~\cref{thm:gk} to $F$. Then, $\fEC(F') \ge \L(F)-(5\negs(F)-2)$.
\end{lemma}
\begin{proof}
By~\cref{thm:gk} the $F'$ obtained is a composition of a read-once formula $H$ over monotone formulas $G_1,\ldots,G_T$ for $T \le 5\negs(F)-2$. 
In addition, by tracing the construction of $F'$ in the proof of~\cref{thm:gk}, it can be inferred that (1) all leaves of $F'$ forms a part of some monotone formula $G_i$ and (2) every leaf in $F$ must appear at least once as a leaf of $F'$. Now,

\begin{align*}
\fEC(F') & \ge \fEC(F',1^n) \\
 & \ge \sum_{i=1}^T \fEC(G_i,1^n) \\ 
 & \ge \sum_{i=1}^T (\L(G_i)-1) && [\text{$G_i$s are monotone}] \\
 & \ge \L(F)-T && [\text{By Property (1) and (2)}]\\
 & \ge \L(F) - (5\negs(F)-2)
\end{align*}
\end{proof}

\cref{energy:formula} holds directly from the following cumbersome but slightly stronger claim. 

\begin{claim}
For any formula $F$,  $\fEC(F) = \Omega\left(\sqrt{\L(F)+\depth(F)^2+\depth(F)}-\depth(F)\right)$.
\end{claim}
\begin{proof}
If $\negs(F) = 0$, then $F$ is monotone and $\fEC(F) = \fEC(F,1^n) = \L(F)-1$. Otherwise, $\negs(F)\ge 1$ and applying~\cref{fml:ub} we have $\fEC(F') \le (5\negs(F)-2)(\fEC(F)+\depth(F)+1)$ and by~\cref{fml:gf:lb} the formula $F'$ obtained satisfy,  $\fEC(F') \ge \L(F)-(5\negs(F)-2)$.

Combining the two bounds on $\fEC(F')$, we have $\fEC(F) \ge  \frac{\L(F)}{5\negs(F)-2}-\depth(F)-2$. Along with~\cref{energy:negs}, we have $$\fEC(F) \ge \max\set{\frac{\L(F)}{5\negs(F)-2}-\depth(F)-2,\negs(F)}$$ 
Let $\alpha$ be the largest possible value such that $\frac{\L(F)}{5\alpha -2} -\depth(F)-2 \ge \alpha$. 
This gives a quadratic equation in $\alpha$ and it can be verified that the  maximizing $\alpha$  is  $\frac{\sqrt{(5\depth(F)+12)^2+20\L(F)} - (5\depth(F)+8)}{10}$.

If $\negs(F)$ is at least $\alpha$, then $\fEC(F) \ge \alpha$. Otherwise, $\fEC(F)$ is lower bounded by $\frac{\L(F)}{5\alpha -2} -\depth(F)-2$ which, by our choice, is at least $\alpha$. Hence in both cases, $$\fEC(F) \ge \alpha = \Omega\left (\sqrt{\L(F)+\depth(F)^2+\depth(F)}-\depth(F) \right ).$$

\end{proof}

\section{Comparison of Lower Bound Techniques for Energy Complexity and a Recent Improvement}\label{ec:lb:dt}
So far, we have seen two techniques to show lower bound for energy complexity one in terms of positive sensitivity (\cref{psens:energy}) for any Boolean function and other in terms of cost of monotone Karchmer-Wigderson game (\cref{kw:lb}) for monotone Boolean functions. In this section, we give a comparison of lower bound techniques for energy complexity with regard to a recent improvement due to Sun~\etal~\cite{SSWX18}.

\cref{energy:dt} says that every Boolean function of small decision tree depth has a  small energy circuit computing it. In the context of proving lower bounds on energy complexity, a natural question to ask is whether a converse of \cref{energy:dt} is true. That is, does a circuit of small energy have a small depth decision tree computing it. More precisely, 
\begin{question} \label{quest}
Is it true that for all Boolean functions $f$, $\DT(f) \le \poly(\EC(f))$ ?
\end{question}

In this context, we give our approach to answer this question using a measure called \textit{max-entropy of a circuit} introduced by Uchizawa \etal~\cite{UDM06}. For a circuit $C$, the max-entropy, denoted by $\entm(C)$ is the logarithm of the number of firing patters of the circuit $C$. As a part of main result, they showed that for any threshold circuit $C$ computing a Boolean function $f$, $\entm(C) \ge \EC_{\calT}(f) -1$. Hence $\entm(C)$ can be seen as yet another measure of energy complexity for threshold circuits. 

Since the same result does not directly extend to circuits over Boolean basis $\calB$, we ask, in a spirit similar to the result of Uchizawa \etal~\cite{UDM06}, if max-entropy is also a measure of energy for Boolean circuits. 
We show in~\cref{entropy-dt} an analogous result for Boolean circuits that for any Boolean function $f$, and a circuit $C$ computing $f$, $\entm(C) = \Omega(\log \DT(f))$. 

\begin{lemma}\label{entropy-dt}
	For a Boolean function $f: \zon \to \zo$, let $C$ be any Boolean circuit computing 
	$f$ having gates computing an arbitrary function of a finite arity. Then, $\DT(f)$ is, asymptotically, at most the number of firing patterns of $C$. Hence $\entm(C) = \Omega(\log \DT(f))$.
\end{lemma}
\begin{proof}
	Let the number of firing patterns of $C$ be $t$ and $\ell$ be the maximum arity of gates 
	in $C$. We show that there exists a decision tree computing $f$ of depth $\ell \cdot t$. 
	Since $\ell$ is a finite constant, $t \ge \DT(f)/\ell = \Omega(\DT(f))$.
	
	Proof is by strong induction on $n$. For $n=1$, $\DT(f) \le 1$ and there must be at 
	least 
	one firing pattern for $C$. Hence $\DT(f) \le \ell \cdot t$.
	Suppose the claim holds for all Boolean functions on $<n$ variables. Let $f$ be an 
	$n$ bit Boolean function computed by a circuit $C$ of size $s$ with gates of fan-in 
	at most $\ell$. For the circuit $C$, let there be $t$ distinct firing patterns $p_1,p_2,\ldots,p_t$ 
	where each $p_i \in \zo^s$. Let $C'$ be the circuit obtained 
	from $C$ by removing all the gates that have the same value in all the firing 
	patterns. Observe that this transformation does not alter the number of firing 
	patterns and let $p_1',p_2',\ldots,p_t'$ be the firing patterns of $C'$. Let $g$ be 
	a gate in $C'$ whose evaluation depends only on input variables. Let $f'$ be the function $f$ after setting the 
	queried variables to the values read. Also set the queried values in $C'$ and 
	evaluate the circuit (as far as possible) to get $C''$ which computes $f'$. Since 
	$f'$ is on $\le n-\ell$ variables, by induction, $\DT(f') \le \ell \times 
	\text{Number of firing patterns of $C''$}$.

	Since the value of gate $g$ is fixed, $C''$ can have at most $t-1$ firing patterns (for otherwise, all the firing patterns have the same value for gate $g$ due to which $g$ would have been removed in $C'$, a contradiction). Hence, $\DT(f) \le \DT(f') +\ell \le \ell \times \text{Number of firing patterns of $C''$} + \ell \le \ell\cdot (t-1) + \ell = \ell \cdot t$. 
\end{proof}
However, this result does not give a meaningful lower bound for energy complexity of $f$. To see this, a circuit with $s$ internal gates and energy $e$ can potentially have $\sum_{i=0}^e \binom{s}{i} \le s^e+1$ firing patterns implying $\entm(C) \le e \log s$.
Now, \cref{entropy-dt} implies that $e = \Omega(\log \DT(f)/\log s)$. However, $\log \DT(f)/\log s = O(1)$ as $\DT(f)= O(s)$.

\cref{entropy-dt} can be seen as constructing a decision tree for $f$, given the firing patterns of a circuit computing $f$. Recently, Sun~\etal~\cite{SSWX18} directly constructed a decision tree of depth $\EC(f)^2$, thereby implying the following.

\begin{theorem}[Sun~\etal~\cite{SSWX18}] \label{eq:sun}
For all Boolean functions $f, \EC(f) \ge \sqrt{\DT(f)}$.
\end{theorem}

This answers~\cref{quest} in affirmative as $\DT(f) \le \EC(f)^2$. The original statement, Theorem 2 of Sun~\etal~\cite{SSWX18}, states that $\EC(f) = \Omega(\sqrt{\DT(f)})$. A careful analysis of their proof reveals that the asymptotic constant is actually $1$.
  In this context, we give two instances where the result of Sun~\etal~\cite{SSWX18} 
can be used to further improve our results from~\cref{sec:psens:lb} and~\cref{sec:kw:lb}.
\begin{itemize}
	\item We showed that $\EC(\land_n) \ge n/3$ based on the measure positive sensitivity (\cref{psens:energy}). But it completely fails to give any non-trivial lower bound for $\EC(\lor_n)$ since $\psens(\lor_n) = 1$. Since $\DT(\lor_n) = n$, by~\cref{eq:sun} this implies that $\EC(\lor_n) \ge \sqrt{n}$ (as observed by Sun~\etal \cite{SSWX18}).
	\item Consider the problem of $\stconn$ which, given a directed graph $G$ on $\binom{n}{2}$ edges and two vertices $s$ and $t$, asks if there is a path from $s$ to $t$ in $G$. It is known that $\KW^+(\stconn) = \Omega(\log^2 n)$~\cite{KW90}. Hence, \cref{kw:lb} implies that $\EC(\stconn) = \Omega(\log^2 n)$. It can be argued that $\DT(\stconn)=\Omega(n^2)$. One way to see this is to observe that connectivity is a non-trivial monotone property of graphs and such properties have decision tree depth of $\Omega(n^2)$~\cite{KSS84}. Hence, by~\cref{eq:sun} we get that $\EC(\stconn) \ge \sqrt{\DT(\stconn)} = \Omega(n)$ which vastly improves what could be inferred via our bound.
\end{itemize}

\section{Discussion and Questions}\label{sec:conc}
Having studied $\EC(f)$ as a Boolean function parameter for different circuit models over the Boolean basis $\calB$, following are some natural questions that are left unanswered.
\begin{itemize}
\item For unbounded fan-in circuits of depth $3$, we showed an energy lower bound of $\Omega(n)$ for parity on $n$ bits (\cref{lb:depththree:raz}). The question here is to extend the same to arbitrary depth unbounded fan-in circuits.
\item For any Boolean formula $F$, we showed a lower bound for $\fEC(F)$ in terms of its size and depth (\cref{energy:formula}). Can we remove the dependence on depth thereby showing that for all Boolean functions $f$, $\fEC(f) = \Omega(\sqrt{\L(f)})$ ?

\end{itemize}

\section*{Acknowledgments}
The authors would like to thank the anonymous reviewers for their constructive comments.

\bibliographystyle{alpha}
\bibliography{master}

\ifthenelse{\equal{\movetoappendix}{1}}{
	\appendix
	\section{Appendix}
	\includecollection{appendix}
} { }

\end{document}